\pdfoutput=1
\documentclass[11pt]{article}
\usepackage[a4paper,lmargin=2.5cm,rmargin=2.5cm,tmargin=4cm,bmargin=3.5cm]{geometry}
\usepackage[utf8x]{inputenc}
\usepackage{dsfont}
\usepackage{amsmath,amssymb,amstext,amsthm,amscd, amsxtra, mathrsfs,bm,slashed}
\usepackage{nicefrac}
\usepackage{enumitem}
\usepackage{array}
\usepackage{cite}
\usepackage{todonotes}
\usepackage{mathtools}
\usepackage[affil-it]{authblk}

 \usepackage{hyperref}
 \hypersetup{
   pdftitle   = {LKMS states},
  pdfauthor  = {M. Gransee, N. Pinamonti, R. Verch},
 }

\theoremstyle{plain}
\newtheorem{lemma}{Lemma}[section]
\newtheorem{prop}[lemma]{Proposition}
\newtheorem{theorem}[lemma]{Theorem}
\newtheorem{coro}[lemma]{Corollary}

\theoremstyle{definition}
\newtheorem{defi}[lemma]{Definition}

\def\ben{\begin{equation}}
\def\een{\end{equation}}
\def\non{\nonumber}

\def\cA{{\mathcal A}}

\def\cC{{\mathcal C}}
\def\cD{{\mathcal D}}

\def\cF{{\mathcal F}}

\def\cI{{\mathcal I}}

\def\cM{{\mathcal M}}

\def\cO{{\mathcal O}}

\def\cP{{\mathcal P}}
\def\cQ{{\mathcal Q}}

\def\cS{{\mathcal S}}
\def\cT{{\mathcal T}}
\def\cU{{\mathcal U}}

\def\IC{{\mathbb C}}
\def\IM{{\mathbb M}}
\def\IN{{\mathbb N}}
\def\IR{{\mathbb R}}

\def\a{\alpha}
\def\b{\beta}
\def\g{\gamma}        \def\G{\Gamma}
\def\d{\delta}        
\def\eps{\varepsilon} 
     
\def\z{z }
    \def\Th{\Theta}
\def\k{\kappa}
\def\l{\lambda}       \def\L{\Lambda}

\def\s{\sigma}		 
    
\def\t{\tau}

\def\c{\chi}
\def\o{\omega}        

\def\fE{{\mathfrak E}}

\def\fR{{\mathfrak R}}
\def\fS{{\mathfrak S}}
\def\fT{{\mathfrak T}}		
		\def\fu{{\mathfrak{u}}}

		\def\fw{{\mathfrak{w}}}

\def\to{\rightarrow}

\def\del{\partial}

\newcommand{\abs}[1]{{\left|{#1}\right|}}			
\newcommand{\eins}{{\mathds{1}}}				
\newcommand{\WF}{{\mathrm{WF}}}				
\newcommand{\wfou}{{\widehat{\fw}}}				
\newcommand{\ufou}{{\widehat{\fu}}}
\newcommand{\diff}{{\mathrm{d}}}				
\newcommand{\dnull}[1]{{\left. \frac{d^{#1}}{d t^{#1}}\right|_{t=0}}} 

\newcommand{\bb}{{\boldsymbol{\b}}}					
\newcommand{\va}{{\underline{\a}}}

\newcommand*{\Bigcdot}{\raisebox{-0.8ex}{\scalebox{2.5}{$\cdot$}}}

\begin{document}
\title{KMS-like Properties of Local Equilibrium States\\ in Quantum Field Theory}
\author[1,4]{Michael Gransee}
\affil[1]{MPI f\"ur Mathematik in den Naturwissenschaften, 04103 Leipzig, Germany}
\author[2,3]{Nicola Pinamonti}
\affil[2]{Dipartimento di Matematica, Universit\`a di Genova, 16146 Genova, Italy}
\affil[3]{INFN, Sezione di Genova, 16146 Genova, Italy}
\author[4]{Rainer Verch}
\affil[4]{Institut f\"ur Theoretische Physik, Universit\"at Leipzig, 04103 Leipzig, Germany}
\date{}
\maketitle

\begin{abstract} A new condition, called ``Local KMS Condition'', characterizing states of a quantum field to which one can
ascribe, at a given spacetime point, a temperature, is introduced in this article.
It will be shown that the Local KMS Condition (LKMS condition) is equivalent to the Local Thermal Equilibrium (LTE) condition, proposed
previously by Buchholz, Ojima and Roos, for states of the quantized scalar Klein-Gordon field that fulfill the analytic
microlocal spectrum condition. Therefore, known examples of states fulfilling the LTE condition 
provide examples of states obeying the LKMS condition with a temperature distribution varying in space and time.
The results extend to the generalized cases of mixed-temperature LKMS and LTE states. The LKMS condition therefore provides
a promising generalization of the KMS condition, which characterizes global thermal equilibrium states with respect to an
inertial time evolution, to states which are globally out of equilibrium but still possess a local temperature distribution.

\end{abstract}

\section{Introduction}
In quantum field theory as well as in quantum statistical mechanics, the characterization of states which have a temperature distribution varying in space and time has proved notoriously difficult, in particular within model-independent frameworks. Already the initial step, consisting of a definition of ``temperature'' of the system at any spacetime point is a non-trivial step. While the KMS (Kubo-Martin-Schwinger) condition serves as both a sufficiently general and suitably restrictive condition of global thermal equilibrium states (in the sense of generalized Gibbs ensembles) for a large class of systems in quantum statistical mechanics and quantum field theory \cite{Emc14,Haa92} a likewise universal concept of ``local'' temperature for such systems does not appear to be at hand. This problem apparently not only besets non-equilibrium situations in (relativistic) quantum field theory, or in quantum systems as such; it is known that, for a wide range of physical systems, there does not seem to exist a universal concept of temperature when the systems are not in thermal equilibrium. We refer to the review \cite{CVJ03} for considerable further discussion on this topic.

One generalization of the KMS condition, the ``non-equilibrium steady states'', has been proposed by D. Ruelle \cite{Rue00}; this class of states corresponds to a medium (described in terms of a generic quantum system) coupled to several heat reservoirs at possibly different temperatures which thereby can attain a temperature distribution that may vary in space, but remains constant in time.

Another approach towards a local temperature concept, intended particularly for application in relativistic quantum field theory, has been proposed by Buchholz, Ojima and Roos \cite{BOR02}. The basic idea here is to 
introduce sets $\cS_q$ of thermal observables of a quantum field theory at any spacetime point $q$ and to compare the expectation values of the thermal observables in a given state with the expectation values of the thermal observables in a global thermal equilibrium state (characterized as a KMS state). The better
the coincidence of the expectation values in the thermal observables, the closer the given state is to a thermal equilibrium state at spacetime point $q$, and may therefore be assigned a temperature at $q$ --- the temperature of the thermal equilibrium state admitting best coincidence of expectation values. We will give a more precise summary of the ``local thermal equilibrium'' (LTE) state concept in Sec.\ 2. Let us remark here that the LTE state concept has some promising features which have been explored in some publications, cf. \cite{Buc03,BS13,SV08,Bah06,Sol10,Sol12,Ver12}. There are, however, some drawbacks. One possible drawback is that the concept
of LTE states very much depends on how the set $\cS_q$ has been chosen. While \cite{BOR02} give very good arguments for their choice of $\cS_q$, it is not clear that this is the only reasonable choice which could be made. Moreover, the definition of $\cS_q$ of \cite{BOR02} does not appear to be related to the time-evolution of the quantum field under consideration, while in contrast, the KMS condition refers explicitly to the time-evolution, and time-like correlations of quantum field observables. That is to say, the LTE concept does not appear as a local generalization of the KMS condition --- the relation to the KMS condition is somewhat indirect
through the use of local thermal observables $\cS_q$ and global thermal equilibrium states of the system for comparison. 

In the present article, we will propose a new class of states, called ``local KMS'' (LKMS) states for the linear, scalar Klein-Gordon field on Minkowski spacetime. This class of states is characterized by a local generalization of the KMS condition,at any spacetime point $q$, at the level of the 2-point functions of the states. 
We will soon give a somewhat more detailed sketch of the LKMS condition to explain to what extent they generalize the KMS condition in a ``local'' sense. Yet, running somewhat ahead, let us already mention that the main result of this article will  consist in showing that, under certain conditions on the states considered, {\it the 
LKMS condition and the LTE condition are equivalent.}  

In order to describe the LKMS condition in more technical terms --- full details and proofs will be given in the main body of the text --- let us suppose that $\phi(f)$, where $f$ denotes test-functions of Schwartz type, are the quantum field operators of the scalar Klein-Gordon field on 4-dimensional Minkowski spacetime. (Basically everything can readily be generalized to $d$-dimensional Minkowski spacetime with $d \ge 4$ upon making some suitable adaptations.) The field operators can be taken to be given in the vacuum representation on their standard Wightman domain; alternatively, they could also be regarded as elements of a $*$-algebra (e.g.\ the Borchers-Uhlmann algebra), initially without a concrete Hilbert space representation. The latter point of view is implicitly adopted when considering expectation values of field operator products in KMS states. Then one can consider a timelike, future-directed, normalized direction vector $e$ of Minkowski spacetime, and we take $\omega_{\beta,e}$ as the KMS-state with respect to the time-direction specified by $e$, at inverse temperature $\beta$. In particular, this implies that the 2-point correlation function $t \mapsto \omega_{\beta,e}(\phi(q - t e /2)\phi(q' + te/2))$ fulfills the KMS condition. Here, we have adopted common practice to denote the distribution 
with respect to the spacetime variables $q$ and $q'$ as if it was a function. Actually, it is an analytic function
as long as $q - t e/2$ and $q' + te/2$ don't coincide or are lightlike related.  Looking only at quantum field operators which, very formally speaking, are localized at the same, coinciding points $q' = q$, one obtains that the 2-point correlation $t \mapsto \omega_{\beta,e}(\phi(q - t e/2)\phi(q + te/2))$ fulfills (a version of) the KMS condition. At this point, one must take some care in expressing the KMS condition because of the distributional nature of that 2-point correlation in $t$ at $t = 0$; however, as we will show in the main body of this article, this difficulty can be overcome.

Then, the basic idea in setting up a local version of the KMS condition that we here propose is to view the KMS condition fulfilled by  $\omega_{\beta,e}(\phi(q - t e/2)\phi(q + te/2))$ with respect to $t$ as a remnant of the full KMS condition at the spacetime point $q$, and to allow the time-direction $e$ as well as the inverse temperature parameter $\beta$ to vary with $q$. In other words, the idea is to say that a (sufficiently regular) state $\omega$ fulfills the local KMS condition at some spacetime point $q$ if there are a time-direction vector $e$ and an inverse temperature $\beta > 0$ such that the two-point correlation function $\omega(\phi(q - t e / 2)\phi(q + t e /2))$ fulfills the same (remnant) KMS-condition with respect to $t$ as does the two-point correlation function $\omega_{\beta,e}(\phi(q - t e /2)\phi(q + t e /2))$ of a proper KMS state $\omega_{\beta,e}$ with respect to the time-direction $e$ at inverse temperature $\beta$.

It turns out that there is some leeway as to how precisely the said idea for an LKMS condition should be implemented. This concerns in particular the behaviour with respect to $t$. The initial thought is to try and keep the LKMS condition as local as possible in $t$ and therefore to base the LKMS criterion on the values of $\omega(\phi(q - t e / 2)\phi(q + t e /2))$ and its $t$-derivatives at $t = 0$. In this spirit, we have given a preliminary definition of the LKMS condition in \cite{Ver12} in the following form. First, abbreviating the 2-point correlation $\omega_{\beta,e}(\phi(q - t e/2)\phi(q + te/2))$ of the KMS state with respect to time direction $e$ and at inverse temperature $\beta$ as 
\begin{align}
 \varphi_q(t) = \omega_{\beta,e}\left(\phi\left(q - \frac{t}{2} e\right)\phi\left(q + \frac{t}{2} e\right)\right) \,
\end{align}
it is not difficult to observe that, as a consequence of the KMS condition, there is a function
\begin{align}
 f_q : S_\beta = \{ t + i\sigma : 0 < \sigma < \beta \} \to \mathbb{C}
\end{align}
which is defined and analytic on the open strip $S_\beta$, and has a continuous extension to the closure $\overline{S}_\beta$
{\it except} at the boundary points with $t = 0$, such that 
\begin{align}
 \lim_{\sigma \to 0}\, (\varphi_q(t) - f_q(t + i\sigma)) = 0 \,, \quad \lim_{\sigma \to \beta} (\varphi_q(-t) - f_q(t + i\sigma)) = 0
\end{align}
holds for all $t \in \mathbb{R}$. More precisely, the previous equations hold at the level of distributions with respect to the variable $t$. In fact, the singularities of $\varphi_q$ and $f_q$ at $t = 0$ compensate in the difference at the boundary of the strip $S_\beta$, so that one has 
\begin{align} \label{KMS-gen}
 \partial_t^n \lim_{\sigma \to 0}\, (\varphi_q(t) - f_q(t + i\sigma)) = 0 \,, \quad \partial_t^n\lim_{\sigma \to \beta} (\varphi_q(-t) - f_q(t + i\sigma)) = 0
\end{align}
for all $n \in \mathbb{N}_0$ and for all $t \in \mathbb{R}$.

In fact, this condition, appropriately interpreted at the level of distributions, {\it together} with the condition that $\varphi_{q + te} = \varphi_q$ for all $t \in \mathbb{R}$, is equivalent to the KMS-condition for the two-point distribution $f_1,f_2 \mapsto \omega(\phi(f_1)\phi(f_2))$ with respect to the time-direction $e$, at inverse temperature $\beta$ \cite{BR97}. The condition $\varphi_{q + t e} = \varphi_{q}$ expresses invariance of condition \eqref{KMS-gen} with respect to the time-shifts along the time-direction given by $e$, meaning in particular that $\beta$ is constant along this time-direction. As we wish to define a condition characterizing states having a temperature at a given point $q$ in spacetime which may vary when $q$ varies, this suggests keeping condition \eqref{KMS-gen} and dropping the condition $\varphi_{q + te} = \varphi_q$. Thus, we were led in \cite{Ver12} to say that a (sufficiently regular) state $\omega$ fulfills the LKMS condition at the spacetime point $q$ if there is a time direction vector $e$ and and inverse temperature $\beta$ (both possibly $q$-dependent) such that, upon writing
\begin{align}
 \psi_q(t) = \omega\left(\phi\left(q - \frac{t}{2} e\right)\phi\left(q + \frac{t}{2} e\right)\right)\,,
\end{align}
there is a function $f_q$ with the properties as specified above such that
\begin{align} \label{LKMS1}
\lim_{t \to 0}\, \partial_t^n \lim_{\sigma \to 0}\, (\psi_q(t) - f_q(t + i\sigma)) = 0 \,, \quad 
\lim_{t \to 0}\, \partial_t^n \lim_{\sigma \to \beta}\, (\psi_q(-t) - f_q(t + i\sigma)) = 0
\end{align}
holds for all $n \in \mathbb{N}_0$.

Such a variant of the LKMS condition appears attractive since it is completely intrinsic in the sense that is uses only properties of the 2-point correlation function of the state $\omega$ in an infinitesimal neighbourhood of the spacetime point $q$, and does not require any states or observables (or limits thereof) for comparison of properties. Moreover, it is apparent in which sense it is a local remnant of the KMS condition. There is also a clear relation to the concept of LTE states. 

To explain that, we recall that a typical choice for $\cS_q$, the local thermal observables at a
spacetime point $q$, is the collection of all balanced derivatives of the Wick-square of the quantum field $\phi$ at $q$. We will explain this in more detail in the main body of this article. As a consequence, if
$\omega$ is an LTE state of infinite order at $q$, then it follows that
\begin{align} \label{diffLTE1}
 \partial_{z^{\mu_1}} \cdots \partial_{z^{\mu_n}} \,\left. \left(\omega(\phi(q + z)\phi(q -z)) -
 \omega_{\beta,e}(\phi(q + z)\phi(q - z)) \right) \right|_{z = 0} = 0
\end{align}
holds for all $n\in\IN_0$, where $z$ is an arbitrary spacelike vector. Assuming sufficient regularity of $\omega$, $\omega(\phi(q)\phi(q')) - \omega_{\beta,e}(\phi(q)\phi(q'))$ is a jointly smooth function of the spacetime points $q$ and $q'$, and thus \eqref{diffLTE1} can be extended to timelike vectors $z$ by applying a polarization formula. Since $\omega_{\beta,e}$ is a KMS state with respect to the time-direction $e$, \eqref{diffLTE1} then implies also that $\omega$ satisfies relation \eqref{LKMS1}. This shows that --- up to some mathematical details which we will properly address in the main part of this work ---  LTE states of infinite order are LKMS states in the sense described above. 

However, we are interested to see under which conditions the LTE and LKMS conditions are equivalent. At this point, it matters to specify the ``sufficient regularity'' of the state $\omega$ (respectively, its two-point function) alluded to before. So far, ``sufficient regularity'' can be expressed as the (implicit) assumption that $\omega$ be a Hadamard state, which is equivalent to saying that its two-point function is of Hadamard form, or, equivalently, that it satisfies the microlocal spectrum condition \cite{Rad96,SVW02}. Making this assumption, one can show that $\omega$ fulfills the LKMS condition if it fulfills the LTE condition (of infinite order) as we have just sketched. However, in order to obtain the reverse conclusion, it turns out that imposing the microlocal spectrum condition on (the two-point function of) $\omega$ is not enough. Stronger regularity on the two-point function of $\omega$ must be imposed, together with a somewhat stronger form of the remnant KMS condition expressed so far in the form of \eqref{LKMS1}. We find that, in order to conclude that the LKMS condition implies the LTE condition, it is sufficient that the two-point function of $\omega$ fulfills the analytic microlocal spectrum condition \cite{SVW02} and that the remnant LKMS condition is strengthened to a remnant form of the relativistic KMS condition previously investigated by Bros and Buchholz \cite{BB94,BB96}. Imposing these stronger conditions, we then obtain as our main result the equivalence of LKMS and LTE conditions (given in versions both for infinite and for finite order). We will also extend the result to mixed-temperature LKMS and LTE states.

As a consequence of the equivalence result, known examples of LTE states with a temperature distribution varying in space and time, such as the hot bang state constructed for the massless Klein-Gordon field in \cite{Buc03}, provide examples for LKMS states with a temperature that varies in space and time. Similarly, results guranteeing the existence of (varying temperature) mixed LTE states \cite{Sol10} also pertain for mixed LKMS states, and likewise, results constraining the existence of LTE states \cite{Buc03} also apply in the case of LKMS states.

The present article is organized as follows.
In Section 2, we discuss basic properties of the quantized Klein-Gordon field on (patches of) Minkowski spacetime, as well as properties of KMS states (with respect to given time-directions) and of their corresponding two-point functions. The LTE (local thermal equilibrium) condition for states of the quantized Klein-Gordon field according to \cite{BOR02} will also be summarized. Section 3 is devoted to definition and discussion of our LKMS condition. This includes various characterizations of the LKMS condition and some results related to their analytic properties. Our main result on the equivalence of LKMS and LTE states, both for LTE states of finite and infinite order, for states fulfilling the analytic microlocal spectrum condition, will also be presented in Section 3. In Section 4, we discuss the generalization of our results to the case of mixed-temperature LTE and LKMS states. The article will be concluded by Summary and Outlook in Section 5. Several auxiliary technical results appear in Appendices to the main text.

\section{Preliminaries}

We consider the quantized uncharged Klein-Gordon field on a globally hyperbolic subregion $\cM\subseteq\IM$ of Minkowski spacetime $\IM=\IR^4$, with Lorentzian pseudo-metric of the form $\eta=\text{diag}(+1,-1,-1,-1)$. The field $\phi(x)$ is regarded as an operator-valued distribution on the space $\cD(\cM)$ of test functions with compact support contained in $\cM$, i.e. formally we have
\ben
\phi(f)=\int \phi(x)f(x) dx,\qquad f\in\cD(\cM).
\een
The field fulfills the Klein-Gordon equation in the sense of distributions, i.e.
\ben
\phi((\Box +m^2)f)=0\quad\forall f\in\cD(\cM),
\een
where $\Box$ denotes the d'Alembertian and $m\geq 0$ is the mass parameter. Furthermore, $\phi(x)$ is regarded as hermitian, i.e. $\phi^*(f)=\phi(\bar{f})$, and to fulfill the Canonical Commutation Relations (CCR), expressed by
\ben
[\phi(f),\phi(g)]=i E(f,g)\eins, \quad f,g\in\cD(\cM),
\label{eq:CCR}
\een
where $[A,B]:=AB-BA$ denotes the commutator and $E$ is the causal propagator, defined as the difference of the advanced minus retarded fundamental solution of the Klein-Gordon equation. Locality (or Einstein causality) is expressed by the fact that $E(f,g)=0$ whenever $f,g\in\cD(\cM)$ have mutually spacelike separated supports. The algebra of local fields is the ${}^*$-algebra $\cA(\cM)$, generated by multiples of $\eins$ and finite sums as well as products of the field operators. If $\cM=\IM$ the algebra of observables $\cA(\IM)$ is stable under the the action of the proper, orthochronous Poincar\'e group  $\mathcal{P}_+^{\uparrow}$, implemented on the field operators by
\ben
 \t_{(\L,a)}(\phi(f))=\phi(f_{(\L,a)}), 
\een
where $f_{(\L,a)}(x)=f(\L^{-1} (x-a))$. If $\cM\subset\IM$ this holds for all $(\L,a)$ in a neighbourhood of the identity in $\cP_+^\uparrow$.

A \emph{state on} $\mathcal{A}(\cM)$ is a normalized positive linear functional\footnote{If $\cA(\cM)$ carries a topology, $\o$ is also required to be continuous with respect to this topology.} $\o:\cA(\cM)\rightarrow \IC$. The $n$-point ``functions'' of a state are distributions $\o_n\in\cD^\prime(\cM^n)$, given by
\ben
\o_{n}(f_1,\ldots,f_n)=\o(\phi(f_1)\cdots\phi(f_n)),\quad n\in\IN.
 \een
It is often convenient to formally write the $n$-point distributions as if they were functions,
\ben
 \o_{n} (x_1,\ldots,x_n):=\o(\phi(x_1)\cdots\phi(x_n)),\quad n\in\IN,
\een
where the $x_j$ are spacetime points. Here, we take the viewpoint of the reconstruction theorem \cite{SW00} and restrict to states which are completely determined by their $n$-point functions $\o_n,n\in\IN$. 

We will call a \emph{quasifree state} any state $\o$ on $\cA(\cM)$ which is completely determined by its two-point function $\o_2$ via
\ben
\o\left(e^{it\phi(f)}\right)=e^{-\frac{1}{2}\o_2(f,f)\cdot t^2},
\een
where the equation is to  be interpreted as equating terms of equal order in $t$.

In what follows we will restrict to so-called Hadamard states. Those are characterized by a specific singular behaviour of their two-point function which mimics that of $\o_2^{\text{vac}}$, the two point function of the unique vacuum state $\o_\text{vac}$ on $\cA(\IM)$. Namely, for every such state $\o$ on $\cA(\cM)$ one has  $\o_2-\left.\o_2^\text{vac}\right|_{\cM\times\cM}\in C^\infty(\cM\times\cM)$, where $\left.\o_2^\text{vac}\right|_{\cM\times\cM}$ is the restriction of $\o_2^\text{vac}\in\cD'(\IR^4\times\IR^4)$ to $\cM\times\cM$. As first recognized by Radzikowski \cite{Rad96}, a Hadamard state $\o$ can be characterized by a certain condition on $\WF(\o_2)$, the wave front set \cite{Hoe90} of its two-point function $\o_2$. We will call a \emph{Hadamard state} a state $\o$ on $\cA(\cM)$ which fulfills
\ben
 \WF(\o_2)=\{(x,x^\prime,k,-k)\in T^\ast\cM^2\backslash\{0\}:x\sim_k x^\prime,k_0>0\},
\label{eq:WFmink}
\een
where $x\sim_k x^\prime$ means that $x$ and $x'$ can be connected by the uniquely defined lightlike geodesic (a straight line in $\cM$) with cotangent vector $k$, while the condition $k_0>0$ requires $k$ to be future directed.
A stronger form of the Hadamard condition, the so-called \emph{analytic microlocal spectrum condition} has been given in \cite{SVW02}. Following this, we will call an \emph{analytic Hadamard state} a state $\o$ on $\cA(\cM)$ which fulfills
\ben
 \WF_A(\o_2)=\{(x,x^\prime,k,-k)\in T^\ast\cM^2\backslash\{0\}:x\sim_k x^\prime,k_0>0\},
\label{eq:WFAmink}
\een
where $\WF_A(\o_2)$ denotes the analytic wavefront set \cite{Hoe90} of $\o_2$. Note that in this case we have $\o_2-\left.\o_2^\text{vac}\right|_{\cM\times\cM}\in C^A(\cM\times\cM)$, where $C^A$ is the class of real-analytic functions. In semiclassical gravity Hadamard states are an indispensable tool when it comes to the regularization of states, since in general there is no preferred state akin to the vacuum in a generic curved spacetime. In particular, they play an important role in the problem of defining an appropriate quantum stress-energy tensor which should replace the classical one in the semiclassical Einstein equations. This has provided some motivation for viewing Hadamard states as physical states for the quantized Klein-Gordon field on curved backgrounds \cite{Wal94}. On the other hand, the characterization of Hadamard states in terms of conditions on the wave front set of $\o_2$ has proved instrumental for quantum field theory on curved spacetimes, cf. \cite{HW01, HW02, BF09, BDH13} and references cited there.

It is well-known that global thermal equilibrium states of infinitely extended quantum systems can be described in a mathematically rigorous manner by means of the so-called \emph{KMS condition}, first considered in \cite{HHW67} (for a detailed discussion of quantum statistical mechanics in the operator-algebraic framework, see \cite{BR97}). In its usual formulation it relies on the existence of a one-parameter group $\a_t$ of automorphisms, the group of time translations. However, in the setting of relativistic QFT, a Lorentz frame is fixed by the choice of a future-directed timelike unit vector $e$ and the latter is interpreted as the ``time-direction'' of this frame. For later use we define the set of time-directions by
\ben
V_+^1:=\{e\in V_+:e^\mu e_\mu=1\},
\een 
where $V_+$ denotes the open forward light cone. Then $$\a_t^{(e)}=\t_{(1,te)},\quad t\in\IR,$$ is the one-parameter group of time evolution on $\cA(\IM)$ with respect to the Lorentz frame whose time-direction is fixed by $e\in V_+^1$. The KMS condition can now be stated as follows:

\begin{defi}
Let $e\in V_+^1$. A state $\o$ on $\cA(\IM)$ is called a \emph{KMS state at inverse temperature} $\b>0$ \emph{with respect to} $\a_t^{(e)}$ (or $(\b,\a_t^{(e)})$-\emph{KMS state}, for short), iff for any $A,B\in\cA(\IM)$ there exists a function $F_{A,B}$, which is defined and analytic on the open strip $S_\b:=\{z\in\IC:0<\Im z<\b\}$, and defined and continuous on $\bar{S}_\b$ with boundary values
\begin{align}
 F_{A,B}(t)&=\o\left(A\a_t^{(e)}(B)\right),\\
 F_{A,B}(t+i\b) &= \o \left(\a_t^{(e)}(B)A\right),\qquad \forall t\in\IR.
\end{align}
\end{defi}

In relativistic QFT a $(\b,\a_t^{(e)})$-KMS state $\o$ is regarded as a thermal equilibrium state at inverse temperature $\b$ with respect to the rest system (or Lorentz frame) specified by some $e\in V_+^1$.  Therefore thermal equilibrium states in relativistic QFT are indicated by both inverse temperature $\b$ and time direction $e$ of the rest system. It is convenient to combine the two quantities into the inverse temperature four-vector $\bb=\b e \in V_+$ so that $\o_\bb$ denotes a $(\b,\a^{(e)}_t)$-KMS state on $\cA(\IM)$. In the following we will call $\o_\bb$ simply a $\bb$-\emph{KMS state}. In the present model we have vanishing chemical potential and can therefore assume that for any given $\bb\in V_+$ there is a unique (gauge-invariant) $\bb$-KMS state $\o_\bb$ on $\cA(\IM)$, i.e. the set $\cC_\bb$ of all $\bb$-KMS states on $\cA(\IM)$ is non-degenerate. This assumption implies that $\o_\bb$ is also invariant under space-time translations. Furthermore, $\o_\bb$ fulfills the following time-clustering property (cf. \cite{BB96}):
\ben
\o_\bb(\phi(f)\a_t^{(e)}\phi(g))\xrightarrow[\abs{t}\to\infty]{} 0\quad \forall f,g\in\cS(\IR^4),
\een
which can be expressed more formally as
\ben
\o_2^\bb(x,y+te)\xrightarrow[\abs{t}\to\infty]{} 0\quad \forall x,y\in\IM.
\label{eq:clust}
\een

The two-point function $\o_2^\bb\in\cD'(\IR^4\times\IR^4)$ of a $\bb$-KMS state $\o_\bb$ on $\cA(\IM)$ in fact is a tempered distribution, $\o_2^\bb\in\cS'(\IR^4\times\IR^4)$. It is given by
\ben
\o_2(f,g)=2\pi\int d^4p \frac{\eps(p_0)\d(p^2-m^2)}{1-e^{-\bb p}}\hat{f}(-p)\hat{g}(p),\quad f,g\in\cS(\IR^4),
\een
where $\cS(\IR^4)$ denotes the space of Schwartz functions on $\IR^4$.

 It has been shown by Buchholz and Bros in \cite{BB94} that for $A,B\in\cA(\IM)$ the correlation functions $F_{A,B}(x):=\o_\bb(A\t_{(1,x)}(B)),x\in\IR^4,$ of $\bb$-KMS states $\o_\bb$ on $\cA(\IM)$ have in fact stronger analyticity properties than that imposed by the KMS condition. These analyticity properties can be viewed as a remnant of the relativistic spectrum condition in the case of thermal equilibrium states, and consequently the term \emph{relativistic KMS condition} was introduced by the authors of \cite{BB94}. 

\subsection{The LTE condition of Buchholz, Ojima and Roos}

In \cite{BOR02} Buchholz, Ojima and Roos developed a method for distinguishing states which are out of equilibrium but locally still have a thermodynamical interpretation. Heuristically speaking, a \emph{local thermal equilibrium (LTE) state} is a state for which one can define local (pointlike) intensive thermal quantities like temperature, pressure and thermal stress-energy which then take the values they would have if the field was in some global thermal equilibrium state. 

The first key step in the analysis of \cite{BOR02} is the construction of spaces $\cQ_q$ of idealized observables (density-like quantities) located at some $q\in\cM$. Those observables are well-defined as quadratic forms in all states with an appropriate high-energy behaviour. From the spaces $\cQ_q$ one then selects certain subspaces $\cS_q\subset\cQ_q$ of local thermal observables $s(q)$. The thermal interpretation of these observables is justified by evaluating them in thermal reference states. The set of these reference states is denoted by $\cC_B$ and consists of mixtures of KMS states $\o_\bb$ on $\cA(\IM)$, with $\bb$ contained in some compact subset $B\subset V_+$. A generic state $\o_B\in\cC_B$ is represented in the form
\ben
\o_B(A)=\int_B \diff \rho(\bb) \o_\bb(A), \quad A\in\cA(\IM),
\label{eq:mixed}
\een
where $\rho$ is a positive normalized measure on $V_+$ with support contained in $B$. In the following we will first restrict to states with \textit{sharply defined} local thermal parameters, corresponding to pointlike-concentrated measures $\rho=\d_\bb$ in (\ref{eq:mixed}). It is thus sufficient to take for the moment as the space of reference states the spaces $\cC_\bb$, defined at the end of the previous section. We will discuss the case of states with``mixed" thermal parameters in Section 4.

For the present model the spaces of thermal observables are defined as the spaces $\cS_q^n$, spanned by the so-called \emph{balanced derivatives of the Wick square up to order} $n$, defined as
\ben
\eth_{\mu_1\ldots\mu_n}:\phi^2:(q):=\lim\limits_{\z\to 0}\del_{z ^{\mu_1}}\ldots\del_{z ^{\mu_n}}\left[\phi(q+\z)\phi(q-\z)-\o_\text{vac}(\phi(q+\z)\phi(q-\z))\cdot\eins\right],
\label{eq:bderiv}
\een
where $\o_\text{vac}$ is the unique vacuum state on $\cA(\IM)$ and the limit is taken along spacelike directions $z $. For the Klein-Gordon field with $m=0$ an easy computation yields
\ben
\fS(\bb):=\o_\bb(\eth_{\mu_1\cdots\mu_n}:\phi^2:(x))=c_n\del_{\bb^{\mu_1}}\ldots\del_{\bb^{\mu_n}}\left(\bb^2\right)^{-1},
\een
where the $c_n$ are some universal numerical constants \cite{BOR02}. This makes clear that the functions $\fS(\bb)$ can be constructed completely out of $\bb$.\footnote{This is also true in the massive case. However, in that case the thermal functions are given by a more involved expression which is analytic in $\bb$. Furthermore, they depend on the choice of a renormalization condition \cite{HW01}} and thus can be viewed as thermal functions corresponding to the micro-observables $s(q)$. Furthermore, due to the invariance of $\o_\bb$ under space-time translations, they are independent of $q$. Note that for odd $n$ the thermal functions are equal to $0$.

The definition of local thermal equilibrium in the sense of \cite{BOR02, Buc03} can now be stated in the case of the Klein-Gordon field as follows:
\begin{defi}
\label{defi:LTE}
 Let $N\in\IN$ and $\omega$ a Hadamard state on $\cA(\cM)$. We say that $\o$ is a \emph{local thermal equilibrium state of order} $N$ \emph{at $q\in\cM$ with sharp inverse temperature vector} $\bb$, or $[\bb,q,N]$-\emph{LTE state} for short, iff there exists a $\bb$-KMS state $\o_{\bb}$ on $\cA(\IM)$, such that
\ben
\o(s(q))=\o_{\bb}(s(q))\quad \forall s(q)\in \cS_q^n,\ n\leq N.
\label{eq:LTE}
\een
We will say that $\o$ is a $[\bb,q]$-LTE state iff it is a $[\bb,q,N]$-LTE state for all $N\in\IN$.
\end{defi}

Of particular interest is the space $\cS_q^2$ which contains (besides the unit $\eins$) two thermal observables which play a prominent role. The first one is $:\phi^2:(q)$, the \textit{Wick square} of $\phi$ at the point $q\in\IM$, defined by
\ben
:\phi^2:(q)=\lim\limits_{\z\to 0}\left[\phi(q+\z)\phi(q-\z)-\o_\text{vac}(\phi(q+\z)\phi(q-\z))\eins\right],
\label{eq:Wick}
\een
where the limit is taken along spacelike directions $z$. This observable is usually regarded as a ''thermometer observable``, which is due to the fact that its evaluation in a $\bb$-KMS state yields for the Klein-Gordon field with $m=0$:\footnote{In the massive case the expression $\o_\bb({:\phi^2:}(q))$ yields a more complicated function of $\b$ which is still monotonously decreasing in $\b$ and contains a possible renormalization freedom coming from $:\phi^2:$ \cite{HW01}.}
\ben
\o_{\bb}(:\phi^2:(q))=\frac{1}{12\b^2}=\frac{k_B^2}{12} T^2.
\een
The other thermal observable contained in $\cS_q^2$ is the second balanced derivative of $:\phi^2:(q)$, which is of special interest since its expectation values $\o_\bb(\eth_{\mu\nu}:\phi^2:(q))$ are, up to a constant, equal to the expectation values of the thermal stress-energy tensor $E_{\mu\nu}(\bb)$. It is a fundamental fact in special relativistic thermodynamics that all relevant macroscopic thermal parameters, in particular the entropy current density, for a (local) equilibrium state can be constructed once the components of $E_{\mu\nu}$ are known \cite[Chapter 4]{Dix78}. For increasing $n$ the spaces $\cS_q^n$ contain more and more elements. Thus, the $[\bb,q,N]$-LTE condition introduces a hierarchy among the local equilibrium states in the following sense: If we successively increase the order $N$ we get an increasingly finer resolution of the thermal properties of this state. For finite $N$ we thus obtain a measure of the deviation of the state $\o$ from complete local thermal equilibrium (which would amount to the case of a $[\bb,q]$-LTE state).

\section{KMS-like properties of sharp-temperature LTE states}

The $\bb$-KMS condition fulfilled by the comparison states $\o_\bb$ can be used as a starting point for an investigation of the analyticity properties of the two-point function $\o_2\in\cD'(\cM\times\cM)$ of a $[\bb,q,N]$-LTE state $\o$ on $\cA(\cM)$.
 We start by noticing that if a state $\o_\bb$ on $\cA(\IM)$ fulfills the $\bb$-KMS condition, then for fixed $f,g\in\cS(\IR^4)$ there exists a function $F_{f,g}$, analytic on $S_\b$, bounded and continuous on $\bar{S_\b}$ such that:
\begin{align}
 F_{f,g}(t)&=\o_\bb\left(\phi(f)\a^{(e)}_t\phi(g)\right),\\[1ex]
 F_{f,g}(t+i\b) &= \o_\bb \left(\a^{(e)}_t\phi(g)\phi(f)\right). 
\end{align}
The $\a^{(e)}_{t}$-invariance of $\o_\bb$, together with the group properties of $\a^{(e)}$, implies:
\begin{align}
 F_{f,g}(t)=\o_\bb\left(\a^{(e)}_{-\frac{t}{2}}\phi(f)\a^{(e)}_\frac{t}{2}\phi(g) \right)\\[1ex]
 F_{f,g}(t+i\b)=\o_\bb\left(\a^{(e)}_\frac{t}{2}\phi(g)\a^{(e)}_{-\frac{t}{2}}\phi(f) \right)
\end{align}

Let $q\in\cM$ arbitrary but fixed, $e\in V_+^1$ a time-direction. In the following, we will denote by $\cI_{q,e}$ an open interval around $0$ resp. by $\cU_q$ an open neighbourhood of $0\in\IR^4$ such that
\begin{align}
 \cI_{q,e}\subset\{t\in\IR:q\pm \frac{t}{2}e\in\cM\}\\[1ex]
\cU_q\subset\{z\in\IR^4:q\pm\frac{z}{2}\in\cM\}.
\end{align}
Now, let $\o$ be a Hadamard state on $\cA(\IM)$ and consider the ``relative-time-variable correlation function'' $f_q\colon\IR\supset \cI_{q,e} \to \IC$,  which is formally defined as
\ben
 f_{q,e}(t):= \o\left(\phi\left(q-\frac{t}{2}e\right)\phi\left(q+\frac{t}{2}e\right)\right),\quad t\in\cI_{q,e},
\label{eq:fq}
\een
where $e\in V_+^1$. Since $\o$ is a Hadamard state, $f_q$ can only be well-defined as a continuous function if restricted to $\cI_q\backslash\{0\}$. Similarly, consider the ``relative-variable correlation function'' $F_q\colon\IR^4\supset \cU_q\to\IC$, formally defined by
\ben
F_q(z ):=\o\left(\phi\left(q-\frac{z}{2}\right)\phi\left(q+\frac{z}{2}\right)\right),\quad z \in\cU_q.
\label{eq:Fq}
\een
This is a well-defined continuous function only as long as restricted to $\cU_q\backslash(\cU_q\cap(\del V_+\cup\del V_-)$, which resembles the restrictions which are met for the vacuum state $\o_\text{vac}$ of a Wightman quantum field theory \cite{SW00}. In this case the \emph{relativistic spectrum condition} plays a crucial role. By using arguments from complex function theory, one is led to the statement that the vacuum expectation values $\o_\text{vac}(\phi(x_1)\cdots\phi(x_n))$ are the \emph{distributional} boundary values of certain holomorphic functions. One could expect that similar statements also hold in the case of thermal equilibrium states. In fact, an axiomatic approach to thermal quantum field theory has been given in \cite{BB96}, wherein the spectrum condition is replaced by the (relativistic) KMS condition. In particular, one has the following relation between the Fourier transform of the  thermal two-point function $\o_2^\bb(x,y )\equiv\o_2^\bb(x-y)$ and that of the commutator function $E(x,y)\equiv E(x-y)$, introduced in eq. (\ref{eq:CCR}):
 \ben
\widehat{\o}_2^\bb(p)=\frac{i\hat{E}(p)}{1-e^{-\bb p}},
\label{eq:KMStwo}
\een
which has to be understood in the sense of distributions. This relation replaces the relation $\widehat{\o}_2^\text{vac}(p)=i\Th(p)\hat{E}(p)$, expressing the relativistic spectrum condition in energy-momentum space.

We first prove a lemma which shows that the (up to now, formal) expressions $f_q$ and $F_q$, defined by eqns. (\ref{eq:fq}) and (\ref{eq:Fq}) can be meaningfully defined in the sense of distributions. 

\begin{lemma}
\label{lemma:pullback}
 Let $\o$ be a Hadamard state on $\cA(\cM)$, $q\in\cM$ and $e\in V_+^1$. Introduce smooth maps $\chi_{q,e}:\cI_{q,e}\to \cM\times\cM$ and $\k_q:\cU_q\to\cM\times\cM$ via
\begin{align}
 \chi_{q,e}(t)&:=\left(q-\frac{t}{2}e,q+\frac{t}{2}e\right),\quad t\in \cI_{q,e},\label{eq:chiq}\\
\k_q(z)&:=\left(q-\frac{1}{2}z,q+\frac{1}{2}z\right),\quad z\in\cU_q.
 \label{eq:kappaq}
\end{align}
Then ${\sf u}_{q,e}:=\chi_{q,e}^* \o_2$ and ${\sf w}_q:=\k_q^* \o_2$, the pullbacks of $\o$ with respect to $\chi_{q,e}$ and $\k_q$, can be defined as distributions in $\cD'(\cI_{q,e})$ and $\cD'(\cU_q)$, such that
\begin{align}
 \WF({\sf u}_{q,e})&\subset\{(0,k)\in\IR^2:k>0\},\label{eq:WFu}\\
  \WF({\sf w}_q)&\subset\left\{(z,p)\in\IR^4\times \IR^4:z\in\del V_+\cup\del V_-,p\in\del V_+(q)\right\}.
  \label{eq:WFw}
\end{align}
If $\o$ is an analytic Hadamard state, the latter relations get replaced by
\begin{align}
  \WF_A({\sf u}_{q,e})&\subset\{(0,k)\in\IR^2:k>0\},\label{eq:WFAu}\\
  \WF_A({\sf w}_q)&\subset\left\{(z,p)\in\IR^4\times \IR^4:z^\mu z_\mu=0,p\in\partial V_+(q)\right\}.\label{eq:WFAw}
\end{align}
\end{lemma}

\begin{proof}
 The proof can be found in Appendix C.
\end{proof}

\begin{defi}
 Let $q\in\cM$ and $\o$ a Hadamard state on $\cA(\cM)$. Then we call ${\sf u}_q\in\cD'(\cI_{q,e})$ the \emph{relative-time-variable two-point function of $\o$ at $q$} and ${\sf w}_q\in\cD'(\cU_q)$ the \emph{relative-variable two-point function of $\o$ at $q$.}
\end{defi}

If $\o$ is a Hadamard state on $\cA(\IM)$, i.e. $\o$ is globally defined, and $q\in\IM$ then $\cI_{q,e}=\IR$ for all $e\in V_+^1$, $\cU_q=\IR^4$, and we denote by $\fu_q\in\cD'(\IR)$ (resp. $\fw_q\in\cD'(\IR^4)$) the relative-time-variable (resp. relative-variable) two-point function of $\o$ at $q$. For a globally hyperbolic subregion $\cM$ of $\IM$, we denote by $\left.\o\right|_\cM$ the restriction of $\o$ to $\cA(\cM)$. Then $\left.\o\right|_\cM$ is a state on $\cA(\cM)$ which fulfills the (local) Hadamard condition, eq. \eqref{eq:WFmink}, and we denote by ${\sf u}_q\in\cD'(\cI_{q,e})$ (resp. ${\sf w}_q\in\cD'(\cU_q)$) the relative-time-variable (resp. relative-variable) two-point function of $\left.\o\right|_\cM$ at $q$. Clearly, ${\sf u}_q$ and ${\sf w}_q$ are the restrictions of $\fu_q$ and $\fw_q$ to $\cI_{q,e}$ and $\cU_q$, i.e. it holds
\begin{align}
{\sf u}_q(h)=\fu_q(h)\quad\forall h\in\cD(\cI_{q,e}),\\
{\sf w}_q(h)=\fw_q(h)\quad \forall h\in\cD(\cU_q).
\end{align} 

Now, if $q\in\IM$ and $\o_\bb$ is a $\bb$-KMS state on $\cA(\IM)$ with $\bb=\b e$,  we have $\fu_q^\bb\in\cS'(\IR)$ and $\fw_q^\bb\in\cS'(\IR^4)$, since the two-point function $\o_2^\bb$ is tempered as well. Furthermore, due to the spacetime-translation invariance of $\o_\bb$, the relative-variable two-point function $\fw_q^\bb$ is independent of the choice of $q\in\IM$ and the same is true for $\fu_q$, i.e. $\fw_q^\bb\equiv\fw_\bb$ and $\fu_q^\bb\equiv\fu_\bb$ for all $q\in\IM$. In fact, for every Hadamard state $\o$ on $\cA(\IM)$ with $\o\circ\t_a=\o$ one has
\ben
\o_2(x,y)=\o_2(0,y-x)=\o_2(x-y,0),
\een
and setting $z :=x-y$ we find
\ben
\fw_q(z )=\o_2\left(q-\frac{z}{2},q+\frac{z}{2}\right)\equiv\fw(z),
\een
to be understood in the sense of distributions. Thus, for translation-invariant states $\o$ on $\cA(\IM)$  we will call ${\fw}$ just the \emph{relative-variable two-point function of $\o$} and it still contains all relevant information on the state (at least for quasifree states). In contrast, a generic Hadamard state $\o$ on $\cA(\cM)$ will not be invariant under spacetime translations and so we have, in general, 
\ben
\o_2(x,y)={\sf w}_{\frac{1}{2}(x+y)}(x-y),\quad x,y\in\cM
\een
or, expressed formally with $q=\frac{1}{2}(x+y)$ and $z =x-y$:
\ben
{\sf w}_q(z )=\o_2\left(q-\frac{z}{2},q+\frac{z}{2}\right),\quad z\in \cU_q.
\een
Thus, in order to know the two-point function $\o_2$ everywhere, one needs to know the  relative-variable two-point functions ${\sf w}_q$ \emph{for all} $q\in\cM$. Thence , there is the possibility of ascribing certain \emph{local} properties (e.g. of thermal nature) to a state in a spacetime region $\cO\subset\cM$ by giving respective conditions only on the relative-variable two-point functions ${\sf w}_q$ with $q\in\cO$.

After these preparations we will now show that for the class of $[\bb,\cO,N]$-states, introduced in Definition \ref{defi:LTE}, one obtains a point-wise remnant of the KMS condition, both in position space and in in momentum space. 

\subsubsection*{Sharp-temperature LTE states of infinite order}

We first show that for sufficiently regular states the $[\bb,q]$-LTE condition of Definition \ref{defi:LTE} can be expressed in terms of the relative-variable two-point function ${\sf w}_q\in\cD'(\cU_q)$, where, as before, $\cU_q$ is an open neighbourhood of $0\in\IR^4$ such that $\cU_q\subset\{z\in\IR^4:q\pm z/2\in\cM\}$.

\begin{lemma}
\label{lemma:LTEequiv}
 Let $q\in\cM$. A Hadamard state $\o$ on $\cA(\cM)$ is a $[\bb,q]$-LTE state if and only if there is a $\bb$-KMS state $\o_{\bb}$ on $\cA(\IM)$, such that 
\ben
[\del_\va ({\sf w}_{\bb}-{\sf w}_q)](0)=0 \quad\forall \va\in\IN_0^4,
\een 
where ${\sf w}_\bb$ is the relative-variable two-point function of $\left.\o_\bb\right|_\cM$.
\end{lemma}

\begin{proof}
In view of the definition (\ref{eq:bderiv}) of the Wick square and its balanced derivatives we have for any Hadamard state $\o$ on $\cA(\cM)$:
\ben
\o(\eth_{\mu_1\ldots\mu_n}:\phi^2:(q))=[\del_{z ^{\mu_1}}\ldots\del_{z ^{\mu_n}}({\sf w}_q-{\sf w}_\text{vac})](0)\quad \forall n\in\IN_0,
\label{eq:bderiv2}
\een
where ${\sf w}_\text{vac}\in\cD'(\cU_q)$ is the relative-variable two-point function of $\left.\o_\text{vac}\right|_\cM$, the restriction to $\cM$ of the unique vacuum state $\o_\text{vac}$ on $\cA(\IM)$. 
${}$\\
\noindent Now let $\o$ be a $[\bb,q]$-LTE state. According to (\ref{eq:bderiv2}) this implies that there is a $\bb$-KMS state $\o_\bb$ on $\cA(\IM)$ such that
\ben
v_1^{\mu_1}\ldots v_n^{\mu_n}[(\del_{z ^{\mu_1}}\ldots\del_{z ^{\mu_n}}({\sf w}_{\bb}-{\sf w}_q))(0)]=0\quad\forall n\in\IN_0;\ v_1,\ldots,v_n\notin \overline{V_+}\cup\overline{V_-}.
\een
Applying Lemma \ref{lemma:symtens} to the symmetric $\binom{0}{n}$-tensors $$\fR_{\mu_1\cdots\mu_n}:=[\del_{z ^{\mu_1}}\ldots\del_{z ^{\mu_n}}({\sf w}_{\bb}-{\sf w}_q)](0),\ n\in\IN_0,$$ it follows that all directional derivatives of ${\sf w}_{\bb}-{\sf w}_q$ vanish at $z =0$.
${}$\\
\noindent Assume conversely that there is a $\bb$-KMS state $\o_\bb$ on $\cA(\IM)$ such that
\ben
[\del_\va ({\sf w}_{\bb}-{\sf w}_q)](0)=0 \quad\forall \va\in\IN_0^4.
\een
In particular we have
\ben
[\del_{z ^\mu_1}\ldots\del_{z ^{\mu_n}}({\sf w}_q-{\sf w}_\text{vac})](0)=[\del_{z ^{\mu_1}}\ldots\del_{z ^{\mu_n}}({\sf w}_{\bb}-{\sf w}_\text{vac})](0)\quad \forall n\in\IN_0,
\een
which implies the ${[\bb,q]}$-LTE condition by (\ref{eq:bderiv2}). This completes the proof.
\end{proof}

\begin{coro}
\label{coro:LTEequiv}
 Let $q\in\cM$. An analytic Hadamard state on $\cA(\cM)$ fulfills the $[\bb,q]$-LTE condition if and only if there exists a $\bb$-KMS state $\o_\bb$ on $\cA(\IM)$ such that $\fw_\bb\in\cS'(\IR^4)$ is the unique extension of the relative-variable two-point function ${\sf w}_q\in\cD'(\IR^4)$ of $\o$ at $q$.
\end{coro}

\begin{proof}
Assume first that ${\sf w}_q$ has the unique extension $\fw_q=\fw_\bb\in\cS'(\IR^4)$. Then restriction to $\cU_q\subset\IR^4$ yields ${\sf w}_q={\sf w}_\bb$ in the sense of distributions which implies the $[\bb,q]$-LTE condition by Lemma \ref{lemma:LTEequiv}.

Assume conversely that $\o$ is a $[\bb,q]$-LTE state. Since  $\o$ and $\o_\bb$ fulfill the analytic Hadamard condition, Lemma \ref{lemma:LTEequiv} implies that ${\sf w}_q={\sf w}_\bb$ in the sense of distributions. Using again the analytic Hadamard property of $\o$ and $\o_\bb$ it follows that  $\fw_q=\fw_\bb$ is the unique  extension (by analyticity) of ${\sf w}_q$ and the proof is complete.
\end{proof}

Note that by the time-clustering property of $\o_{\bb}$, eq. (\ref{eq:clust}), this also implies that
\ben
\fw_q(te)=\fw_{\bb}(te)=\o_2^{\bb}(q,q+t e)\xrightarrow[\abs{t}\to\infty]{} 0,
\een
which means that $\o$ fulfills a ``localized" version of the time-clustering property.

The next Proposition shows that if $\o$ is a $[\bb,q]$-LTE state on $\cA(\IM)$, where $\bb=\b e$ with $\b>0$ and $e\in V_+^1$,  the relative-time-variable two-point function ${\sf u}_q\in\cD'(\cI_{q,e})$, defined in Lemma \ref{lemma:pullback}, has certain analyticity and periodicity properties which are inherited from the $\bb$-KMS condition fulfilled by the comparison state $\o_\bb$.  Furthermore, these properties can be expressed as a remnant of the $\bb$-KMS condition in momentum space.

\begin{prop}
\label{prop:KMSrem1}
  Let $q\in\cM$, let $\cI_{q,e}$ be an open interval around $0\in\IR$ with $\cI_{q,e}\subset\{t\in\IR:q\pm te/2 \in\cM\}$,  and let $\o$ be an analytic Hadamard state on $\cA(\cM)$ which fulfills the $[\bb,q]$-LTE condition. Then the following hold:
\begin{itemize}  
 \item[(i)] There exists a complex function $f_q$ with the following properties:
\begin{itemize}
\item[(a)] $f_q$ is defined and holomorphic on $S_{q}=\{\z\in\IC:0<\Im\z<\b\}$.
\item[(b)] For all compact $K\subset(0,\b)$ there exist constants $C_K>0$ and $N_K\in\IN_0$ such that
\ben
\abs{f_{q}(\z)}\leq C_K (1+\abs{\z})^{N_K},\quad \forall \Im \z\in K.
\een
\item[(c)] We have:
\begin{align}
 \int dt f_q(t+i\s )h(t)&\xrightarrow[\s\to 0^+]{}{\sf u}_q(h)\quad \forall h\in\cD(\cI_{q,e}),\\[1ex]
 \int dt f_q(t+i(\b-\eta))h(t)&\xrightarrow[\eta\to 0^+]{}{\sf u}_q(h^-)\quad \forall h\in\cD(\cI_{q,e}),
\end{align}
where $h^-(t):=h(-t)$.
\end{itemize} 
\item[(ii)] The relative time-variable two-point function ${\sf u}_q$ of $\o$ at $q$ has a unique extension $\fu_q\in\cS'(\IR)$ such that 
\ben
e^{\b k}\ufou_q(-k)=\ufou_q(k).
\label{eq:KMSrem1fou}
\een
in the sense of distributions.
\end{itemize}
\end{prop}

\begin{proof}
By Corollary \ref{coro:LTEequiv} the $[\bb,q]$-LTE condition in particular implies that the relative-time-variable two-point function ${\sf u}_q\in\cD'(\cI_{q,e})$ has the unique extension $\fu_q=\fu_\bb\in\cS'(\IR^4)$, which by eq. (\ref{eq:KMSpull}) in Appendix B yields
\ben
\ufou_q(-k)=\ufou_{\bb}(-k)=\frac{i\widehat{\fE}(k)}{e^{\b k}-1},
\een
where $\widehat{\fE}(k)$ is given for  $m=0$ and $m>0$ by eq. (\ref{eq:commpull0}) and eq. (\ref{eq:commpullm}) in Appendix B. This immediately implies $(ii)$.
Now, defining $\g_q:=\{\s\in\IR:e^{\s k}\ufou_q(-k)\in\cS'(\IR)\}$, we have $0\in\g_q$ and $\b\in\g_q$, since $e^{\b k}\ufou_q(-k)=\ufou_q(k)\in \cS'(\IR)$. The set $\g_q$ is convex \cite[Thm. 2-5]{SW00} and it follows that $\l\b\in\g_q$ for all $\l\in[0,1]$. In other words, $e^{\s k}\ufou_q(-k)\in\cS'(\IR)$ for all $\s \in[0,\b]$. From the first part of Thm 7.4.2 in \cite{Hoe90} it then follows that there exists an holomorphic function $f_q:S_q\rightarrow\IC$ such that $f_q(t+i\s)=\cF[e^{\s k}\ufou_q(-k)](t)$ and property (ii) above holds. This immediately implies that for every fixed $\s_0\in(0,\b)$ we have $f_q(\Bigcdot+i\s_0)\in\cS^\prime(\IR)$. Since the Fourier transform $\cF:\cS'(\IR) \rightarrow\cS'(\IR)$ is weakly continuous and we know that $e^{\s k}\ufou_q(-k) \xrightarrow[\s\to 0^+]{} \ufou_q(-k)$ in $\cS'(\IR)$, we find
\ben
f_q(t+i\s)=\cF[e^{\s k}\ufou_q(-k)] \xrightarrow[\s\to 0^+]{} \cF[\ufou_q(-k)]=\fu_q(t)
\een
in the sense of distributions. Now consider the function $g_q$, defined by $g_q(t+i\eta):=f_q(t+i(\b-\eta))$, which is as well defined and holomorphic on $S_q$. Then analogous arguments as for $f_q$, together with (\ref{eq:KMSfou}) yield
\ben
 g_q(t+i\eta)=\cF[e^{(\b-\eta)k}\ufou_q(-k)] \xrightarrow[\eta\to 0^+]{} \cF[e^{\b k}\ufou_q(-k)]=\cF[\ufou_q(k)]=\fu_q(-t)
\een
in the sense of distributions. Restriction to test functions $h\in\cD(\cI_{q,e})$ completes the proof of $(ii)\Rightarrow (i)$ and thus the proof of the Proposition.

\end{proof}

It is immediately clear that properties $(i)$ and $(ii)$ in Prop. \ref{prop:KMSrem1} are not equivalent to the $[\bb,q]$-LTE condition. Namely, these (equivalent) conditions are too weak to imply that the relative-variable two-point function ${\sf w}_q\in\cD'(\cU_q)$ has the unique extension $\fw_q=\fw_\bb\in\cS'(\IR^4)$. In this respect, the $[\bb,q]$-LTE condition seems to impose stronger constraints on the relative variable two-point function.
The next lemma shows that the (relative-variable) two-point function $\fw_\bb$ of a $\bb$-KMS state $\o_\bb$ has in fact stronger analyticity properties than those imposed by Proposition \ref{prop:KMSrem1}. Those are remnants of the relativistic KMS condition of Buchholz and Bros \cite{BB94}.

\begin{lemma}
\label{lemma:KMSrem4}
 Let $\o_\bb$ be a $\bb$-KMS state on $\cA(\IM)$. Then there exists a complex function $F_\bb$ with the following properties:
\begin{itemize}
 \item[(i)] $F_\bb$ is defined and holomorphic in the (flat) tube $\cT_\bb=\{\zeta\in\IC^4:\Im \z=\s e, 0<\s<\b\}$.
\item[(ii)] For all compact $K\subset(0,\b)$ there exist constants $C_K>0$ and $N_K\in\IN_0$ such that
\ben
\abs{F_\bb(z +i\s e)}\leq C_K(1+\abs{z +i\s e})^{N_K},\quad z \in\IR^4,\s\in K.
\een  
\item[(iii)] We have, in the sense of distributions:
\begin{align}
 F_\bb(z +i\s e)&\xrightarrow[\s \to 0^+]{} \fw_\bb(z ),\\
 F_\bb(z +i(\b-\eta) e)&\xrightarrow[\eta \to 0^+]{} \fw_\bb(-z ),
\end{align}
where $\fw_\bb\in\cS'(\IR^4)$ is the relative-variable two-point function of $\o_\bb$.
\end{itemize}
\end{lemma}

\begin{proof}
 The proof can be found in Appendix C.
\end{proof}

The above properties of $\fw_\bb$ are not equivalent to the $\bb$-KMS condition. Since one has to choose some (arbitrary) point $q\in\IM$ in defining $\fw_\bb$, Lemma \ref{lemma:KMSrem4} yields only some \emph{local} properties of the state $\o_\bb$ (i.e. properties at $q$) which itself do not imply space-time translation invariance of $\o_\bb$. Based on this observation, we propose the following generalization of the $\bb$-KMS condition:

\begin{defi}
\label{defi:LKMS}
 Let  $q\in\cM$, let $\cU_q$ be an open neighbourhood of $0\in\IR^4$ with $\cU_q\subset\{z\in\IR^4:q\pm z/2\in\cM\}$, and let $\o$ be an analytic Hadamard state on $\cA(\cM)$. We say that $\o$ fulfills the \emph{local KMS condition at $q$ with respect to $\bb$}, or $[\bb,q]$-\emph{LKMS condition} for short, iff there exists a $\bb=\b e\in V_+$ and a complex function $F_q$ with the following properties:
\begin{itemize}
\item[(i)] $F_q$ is  defined and holomorphic in the (flat) tube $\cT_q=\{\zeta\in\IC^4:\Im \zeta=\s e, 0<\s<\b\}$.
\item[(ii)] For all compact $K\subset(0,\b)$ there exist constants $C_K>0$ and $N_K\in\IN_0$ such that
\ben
\abs{F_q(z +i\s e)}\leq C_K(1+\abs{z +i\s e})^{N_K},\quad z \in\IR^4,\s\in K.
\label{eq:LKMS1}
\een
\item[(iii)] We have: 
\begin{align}
 \int d^4 z\ F_q(z +i\s e) h(z)&\xrightarrow[\s\to 0^+]{}{\sf w}_q(h)\quad \forall h\in\cD(\cU_q),
\label{eq:LKMS2}\\[1ex]
 \int d^4 z F_q(z +i(\b-\eta)e)h(z)&\xrightarrow[\eta\to 0^+]{}{\sf w}_q(h^-)\quad \forall h\in\cD(\cU_q),
\label{eq:LKMS3}
\end{align}
where $h^-(z):=h(-z)$.
\item[(iv)] We have
\ben
F_q((t+i\s)e)\xrightarrow[\abs{t}\to\infty]{} 0\quad\forall \s\in(0,\b)
\label{eq:LKMS4}
\een
\end{itemize}
Let $\cO\subset\IM$. We say that $\o$ fulfills the $[\bb,\cO]$-\emph{LKMS condition}, iff there exists a continuous (resp. smooth, if $\cO$ is open) map $\bb:\cO\to B\subset V_+$, such that $\o$ fulfills the $[\bb,q]$-LKMS condition for all $q\in\cO$, with $\bb=\bb(q)$. 
\end{defi}

We are now able to show that this condition indeed is equivalent to the $[\bb,q]$-LTE condition and can as well be formulated in momentum space.

\begin{theorem}
\label{theorem:LKMS}
  Let $q\in\IM$ and $\o$ an analytic Hadamard state on $\cA(\cM)$. Then the following are equivalent:
\begin{itemize}
 \item[(i)] $\o$ is a $[\bb,q]$-LTE state.
\item[(ii)] $\o$ fulfills the $[\bb,q]$-LKMS condition.
\item[(iii)] There exists a $\bb=\b e\in V_+$ such that the relative-variable two-point function ${\sf w}_q$ has an extension $\fw_q\in\cS'(\IR^4)$ with
\ben
\fw_q(te)\xrightarrow[\abs{t}\to\infty]{} 0
\label{eq:clustrem}
\een
and the following holds in the sense of distributions:
\ben
e^{\bb p}\wfou_q(-p)=\wfou_q(p).
\label{eq:LKMSfou}
\een
\end{itemize}
\end{theorem}

\begin{proof}
We show $(i)\Rightarrow(ii)\Rightarrow(iii)\Rightarrow(i)$.

Let $\o$ be a $[\bb,q]$-LTE state on $\cA(\IM)$. Then, by Corollary \ref{coro:LTEequiv}, ${\sf w}_q\in\cD'(\cU_q)$ has the unique extension $\fw_q=\fw_\bb$. Applying Lemma \ref{lemma:KMSrem4} to the $\bb$-KMS state $\o_{\bb}$ and subsequent restriction to test functions $h\in\cD'(\cU_q)$ yields $(ii)$, with $F_q\equiv F_{\bb}$.

Assume that $\o$ fulfills the $[\bb,q]$-LKMS condition. By \cite[Thm. 7.4.2]{Hoe90} there exists a distribution $\hat{T}_q\in\cS'(\IR^4)$ such that $e^{\s e p}\hat{T}_q\in\cS'(\IR^4)$ for all $0\leq\s\leq \b$ and $F_q(z +i\s e)=\cF[e^{\s e p}\hat{T}_q(p)](z )$. Since $\cF\colon\cS'\to\cS'$ is weakly continuous and $e^{\s e p}\hat{T}_q\xrightarrow[\s\to 0^+]{} \hat{T}_q$ in the weak sense, it follows that
\ben
F_q(z +i\s e)\xrightarrow[\s \to 0^+]{}\cF[\hat{T}_q(p)](z )=T_q(-z).
\een
By eq. \eqref{eq:LKMS2} this yields 
\ben 
{\sf w}_q(h)=T_q(h^-)=T_q^{-}(h)\quad \forall h\in\cD(\cU_q),
\label{eq:Tq}
\een
 i.e. $T_q^{-} \in\cS'(\IR^4)$ is an extension of ${\sf w}_q$. Now set $G_q(z +i\s e):=F_q(z +i(\b-\s)e)$, which as well is defined and holomorphic on $\cT_q$. Analogous arguments as above yield
\ben
G_q(z +i\s e)=\cF[e^{(\b-\s) e p}\hat{T}_q(p)](z )\xrightarrow[\s \to 0^+]{}\cF[e^{\bb p}\hat{T}_q(p)](z)
\een
 in the sense of distributions, and restriction to test functions $h\in\cD(\cU_q)$, together with eq. \eqref{eq:LKMS3} and eq. \eqref{eq:Tq} gives
\ben
\left. \cF^{-1}\left[e^{\bb p}\hat{T}_q\right]\right|_{\cU_q}={\sf w}_q+{\sf R}_q=\left.T_q^-\right|_{\cU_q}.
\een
Writing $\fw_q:=T_q^-$ for the extension of ${\sf w}_q$, this implies
\ben
e^{\bb p}\wfou_q(-p)=\wfou_q(p).
\een
That eq. \eqref{eq:clustrem} is fulfilled follows directly from eq. \eqref{eq:LKMS4} by which the proof of $(ii)\Rightarrow(iii)$ is complete.

Now assume that $(iii)$ holds. From (\ref{eq:LKMSfou}), together with the canonical commutation relations (\ref{eq:CCR}), it follows that
\begin{align}
i\hat{E}(p)=\wfou_q(p)-\wfou_q(-p)=(e^{\bb p}-1)\wfou_q(-p).
\end{align}
As discussed before, the distributional solution $\fw_q(-p)$ of the latter equation is not unique:
\ben
\wfou_q(-p)=\frac{i\hat{E}(p)}{e^{\bb p}-1}+c\delta(p),\quad c\in\IR.
\een 
Thus, by (\ref{eq:KMStwo}) there exists a $\bb$-KMS state state $\o_\bb$, such that
\ben
\wfou_q(p)=\wfou_{\bb}(p)+c\delta(p),
\een
where we have used $\d(-p)=\d(p)$. But in view of the clustering property, eq. (\ref{eq:clustrem}), we find that $\fw_q$ cannot contain a constant term, i.e. $c=0$ and we have
\ben
\fw_q=\fw_{\bb}
\een
in the sense of distributions. Since this is equivalent to $(i)$ by Corollary \ref{coro:LTEequiv}, the proof is complete.
 \end{proof}

As already mentioned in the introduction, this result in particular implies that known examples of $[\bb,q]$-LTE states are $[\bb,q]$-LKMS states which guarantees existence of the latter. For example, as discussed in \cite{BOR02,Buc03}, in the case of a massless field the \emph{hot-bang states} $\o_\text{hb}$, whose two-point functions a formally given by
\ben
\o_2^\text{hb}(x,y)=\frac{1}{(2\pi)^3}\int d^4p \frac{\eps(p_0)\d(p^2-m^2)}{1-e^{-A(x+y)^\mu p_\mu}}\quad \forall x,y\in V_+, 
\een
with constant $A>0$, constitute examples of $[\bb,V_+]$-LTE states. Therefore, as a consequence of Theorem \ref{theorem:LKMS}, they fulfill the $[\bb,q]$-LKMS condition for all $q\in V_+$, with $\bb\equiv\bb(q)=A\cdot q$.
Apart from these examples we immediately obtain that the restriction $\left.\o_\bb\right|_\cO$ of a $\bb$-KMS state $\o_\bb$ on $\cA(\IM)$ to a spacetime region $\cO$ fulfills the $[\bb,q]$-LKMS condition with $\bb=\text{const.}$ for all $q\in\cO$. Note, that these restricted states do not fulfill the KMS condition. Thus, we have the following Corollary of Theorem \ref{theorem:LKMS}: 

\begin{coro}
\label{coro:exLKMS}
For any given $q\in\cM$ and $\tilde{\bb}\in V_+$ there exist states $\o_{\tilde{\bb},q}$ on $\cA(\cM)$ which fulfill the $[\tilde{\bb},q]$-LKMS condition.
\end{coro}
\subsubsection*{Finite-order LTE-states}

Similar to the case of infinite-order LTE states one can express the $[\bb,q,N]$-LTE condition for analytic Hadamard states in terms of the relative-variable two-point function ${\sf w}_q\in\cD'(\cU_q)$:

\begin{lemma}
\label{lemma:LTENequiv}
 Let $q\in\IM$ and $N\in\IN_0$. A Hadamard state $\o$ on $\cA(\cM)$ is a $[\bb,q,N]$-LTE state if and only if there exists a $\bb$-KMS state $\o_{\bb}$ such that
\ben 
[\del_\va ({\sf w}_{\bb}-{\sf w}_q)](0)=0 \quad\forall \va\in\{\va'\in\IN_0^4:\abs{\va'}\leq N\}.
\label{eq:LTENequiv}
\een
\end{lemma}

\begin{proof}
Assume that $\o$ is a $[\bb,q,N]$-LTE state on $\cA(\IM)$, which in view of (\ref{eq:bderiv2}) yields
\ben
v_1^{\mu_1}\ldots v_n^{\mu_n}[(\del_{z ^{\mu_1}}\ldots\del_{z ^{\mu_n}}({\sf w}_{\bb}-{\sf w}_q))(0)]=0\quad\forall n\leq N;\ v_1,\ldots,v_n\notin \overline{V_+}\cup\overline{V_-}.
\een
Applying Lemma \ref{lemma:symtens} to the symmetric $\binom{0}{n}$-tensors $$\fR_{\mu_1\cdots\mu_n}:=[\del_{z ^{\mu_1}}\ldots\del_{z ^{\mu_n}}({\sf w}_{\bb}-{\sf w}_q)](0),\ n\leq N,$$ it follows that 
\ben 
[\del_\va ({\sf w}_{\bb}-{\sf w}_q)](0)=0 \quad\forall \va\in\{\va'\in\IN_0^4:\abs{\va'}\leq N\}.
\een
Assume conversely that there exists a $\bb$-KMS state $\o_{\bb}$ on $\cA(\IM)$ such that the latter equation holds. In particular we have
\ben
[\del_{z ^{\mu_1}}\ldots\del_{z ^{\mu_n}}({\sf w}_q-{\sf w}_\text{vac})](0)=[\del_{z ^{\mu_1}}\ldots\del_{z ^{\mu_n}}({\sf w}_{\bb}-{\sf w}_\text{vac})](0) \quad\forall n\leq N,
\een
which implies the $[\bb,q,N]$-LTE condition by the definition of the balanced derivatives of the Wick square and by eq. (\ref{eq:bderiv2}).
\end{proof}

\begin{coro}
\label{coro:LTENequiv}
Let $q\in\cM$ and $N\in\IN_0$. An analytic Hadamard state $\o$ on $\cA(\cM)$ is a $[\bb,q,N]$-LTE state if and only if there exists a $\bb$-KMS state $\o_\bb$ on $\cA(\IM)$ and symmetric distribution ${\sf R}_q\in\cD'(\cU_q)$ with $\WF_A({\sf R}_q)=\emptyset$ and
\ben
[\del^\va {\sf R}_q(0)]=0 \quad\forall \va\in\{\va'\in\IN_0^4:\abs{\va'}\leq N\}
\een
such that $\fw_\bb\in\cS'(\IR^4)$ is the unique extension of ${\sf w}_q+{\sf R}_q\in\cD'(\cU_q)$, where ${\sf w}_q$ is the relative-variable two-point function  of $\o$ at $q$.
\end{coro}

\begin{proof}
Assume first that there exists a $\bb$-KMS state $\o_\bb$ on $\cA(\IM)$ and a ${\sf R}_q\in\cD'(\cU_q)$ with the above properties such that $\fw_\bb$ is the unique extension of ${\sf w}_q+{\sf R}_q$. In particular, by restriction to $\cU_q\subset\IR^4$ , we have ${\sf w}_q={\sf w}_\bb+{\sf R}_q$ in the sense of distributions. Therefore, 
\ben 
[\del_\va ({\sf w}_{\bb}-{\sf w}_q)](0)=[\del_\va {\sf R}_q](0)=0 \quad\forall \va\in\{\va'\in\IN_0^4:\abs{\va'}\leq N\},
\een
which implies the $[\bb,q,N]$-LTE condition by Lemma \ref{lemma:LTENequiv}.
${}$\\
\noindent Assume conversely that $\o$ is a $[\bb,q,N]$-LTE state on $\cA(\cM)$ and define ${\sf R}_q\in\cD'(\cU_q)$ by $${\sf R}_q(z):=({\sf w}_\bb-{\sf w}_q)(z)\ \forall z\in\cU_q.$$ Since $\o$ and $\o_\bb$ fulfill the analytic Hadamard condition it follows that $\WF_A({\sf R}_q)=\emptyset$ and , since by eq. \eqref{eq:CCR} the antisymmetric parts of ${\sf w}_\bb$ and ${\sf w}_q$ cancel each other, ${\sf R}_q$ is symmetric, ${\sf R}_q(-z)={\sf R}_q(z)\ \forall z\in\cU_q$. Furthermore, by eq. \eqref{eq:LTENequiv}, it follows that
\ben
[\del^\va {\sf R}_q(0)]=[\del_\va ({\sf w}_{\bb}-{\sf w}_q)](0)=0 \quad\forall \va\in\{\va'\in\IN_0^4:\abs{\va'}\leq N\}.
\een
Since ${\sf w}_\bb={\sf w}_q+{\sf R}_q$ in the sense of distributions it follows from the analytic Hadamard condition on $\o $ and $\o_\bb$ as well as from the analyticity of ${\sf R}_q$ that 
\ben
\WF_A({\sf w}_q+{\sf R}_q)=\WF_A({\sf w}_q)=\WF_A({\sf w}_\bb).
\een
Thus, we obtain that $\fw_\bb\in\cS'(\IR^4)$ is the unique extension (by analyticity) of ${\sf w}_q+{\sf R}_q\in\cD'(\cU_q)$ and the proof is complete. 
\end{proof}

\begin{defi}
\label{defi:LTEN}
Let  $q\in\cM$,  let $\cU_q$ be an open neighbourhood of $0\in\IR^4$ with $\cU_q\subset\{z\in\IR^4:q\pm z/2\in\cM\}$, and let $\o$ be an analytic Hadamard state on $\cA(\cM)$. We say that $\o$ fulfills the \emph{local KMS condition at $q$ with respect to $\bb$ at order $N\in\IN$}, or $[\bb,q,N]$-\emph{LKMS condition} for short, iff here exists a $\bb=\b e\in V_+$ and a complex function $F_q$ with the following properties:
\begin{itemize}
 \item[(i)] $F_q$ is defined and holomorphic in the (flat) tube $\cT_q=\{\zeta\in\IC^4:\Im \zeta=\s e, 0<\s<\b\}$.
 \item[(ii)] For all compact $K\subset(0,\b)$ there exist constants $N_K\in\IN$ and $C_K>0$ such that
\ben
 \abs{F_q(z+i\s e)}\leq C_K(1+\abs{z +i\s e})^{N_K},\qquad \forall\ \s\in K.
 \label{eq:LKMSN1} 
\een
 \item[(iii)] There exists a symmetric distribution ${\sf R}_q\in\cD'(\cU_q)$ with $\WF_A({\sf R}_q)=\emptyset$ and
\begin{align} 
[\del^\va {\sf R}_q(0)]=0 &\quad\forall \va\in\{\va'\in\IN_0^4:\abs{\va'}\leq N\},\label{eq:Rq2}
\end{align}
such that 
\begin{align}
 \int d^4 z\ F_q(z +i\s e)h(z)&\xrightarrow[\s\to 0^+]{}({\sf w}_q+{\sf R}_q)(h) \quad\forall h\in 
\cD(\cU_q)
 \label{eq:LKMSN2}
\\[1ex]
\int d^4z\ F_q(z +i(\b-\eta)e)h(z)&\xrightarrow[\eta\to 0^+]{}({\sf w}_q+{\sf R}_q)(h^-)\quad\forall h\in 
\cD(\cU_q).
 \label{eq:LKMSN3}
\end{align}
\item[(iv)] We have
\ben
F_q((t+i\s)e)\xrightarrow[\abs{t}\to\infty]{} 0\quad\forall \s\in(0,\b).
\label{eq:LKMSN4}
\een
\end{itemize}
\end{defi}

We now prove a generalization of Theorem \ref{theorem:LKMS}.

\begin{theorem}
 Let $q\in\IM$, $\o$ an analytic Hadamard state on $\cA(\IM)$. and denote by ${\sf w}_q\in\cD'(\cU_q)$ the relative-variable two-point function of $\o$ at $q$. Then the following are equivalent:
\begin{itemize}
 \item[(i)] $\o$ is a $[\bb,q,N]$-LTE state. 
\item[(ii)] $\o$ fulfills the $[\bb,q,N]$-LKMS condition.
\item[(iii)]  There exists a $\bb\in V_+$ and a symmetric distribution ${\sf R}\in\cD'(\cU_q)$ with $\WF_A({\sf R}_q)=\emptyset$ and
\ben
[\del^\va {\sf R}_q(0)]=0 \quad\forall \va\in\{\va'\in\IN_0^4:\abs{\va'}\leq N\},
\een
 such that ${\sf w}_q+{\sf R}_q$ has an extension $\fT_q\in\cS'(\IR^4)$ with
 \ben
 \fT_q(te)\xrightarrow[\abs{t}\to\infty]{} 0,\label{eq:clustrem2} 
 \een
 and the following holds in the sense of distributions:
\ben
e^{\bb p}\hat{\fT_q}(-p)=\hat{\fT}_q(p).
\label{eq:LKMSNfou}
\een
\end{itemize}
\end{theorem}

\begin{proof}
 We show $(i)\Rightarrow (ii)\Rightarrow (iii)\Rightarrow (i)$.

Let $\o$ be a $[\bb,q,N]$-LTE state on $\cA(\cM)$. Then, by Corollary \ref{lemma:LTENequiv}, ${\sf w}_q+{\sf R}_q\in\cD'(\cU_q)$ has the unique extension $\fw_\bb\in\cS'(\IR^4)$. Applying Lemma \ref{lemma:KMSrem4} to the $\bb$-KMS state $\o_{\bb}$ and subsequent restriction to test functions $h\in\cD(\cU_q)$ yields $(ii)$, with $F_q\equiv F_{\bb}$.

Now assume that $(ii)$ holds. By \cite[Thm. 7.4.2]{Hoe90} there exists a distribution $\hat{\fT}_q\in\cS'(\IR^4)$ such that $e^{\s e p}\hat{\fT}_q\in\cS'(\IR^4)$ for all $0\leq\s\leq \b$ and $F_q(z +i\s e)=\cF[e^{\s e p}\hat{\fT}_q(p)](z )$. Since $\cF\colon\cS'\to\cS'$ is weakly continuous and $e^{\s e p}\hat{\fT}_q\xrightarrow[\s\to 0^+]{} \hat{\fT}_q$ in the weak sense, it follows that
\ben
F_q(z +i\s e)\xrightarrow[\s \to 0^+]{}\cF[\hat{\fT}_q(p)](z)=\fT_q(-z).
\een
By eq. \eqref{eq:LKMSN2} this yields
\ben
\fT_q(h^-)=\fT_q^-(h)=({\sf w}_q+{\sf R}_q)(h)\quad  \forall h\in\cD(\cU_q),
\label{eq:TqN}
\een 
i.e. $\fT_q^-\in\cS'(\IR^4)$ is an extension of ${\sf w}_q+{\sf R}_q$. Now set $G_q(z +i\s e):=F_q(z +i(\b-\s))$, which as well is defined and holomorphic in $\cT_q$. By analogous arguments as above we have
\ben
G_q(z +i\s e)=\cF[e^{(\b-\s) e p}\hat{\fT}_q(p)](z )\xrightarrow[\s \to 0^+]{}\cF[e^{\bb p}\hat{\fT}_q(p)](z)
\een
in the sense of distributions. Restriction to test functions $h\in\cD(\cU_q)$, together with \eqref{eq:LKMSN3} and eq. (\ref{eq:TqN})gives
\ben
\left. \cF^{-1}\left[e^{\bb p}\hat{\fT}_q\right]\right|_{\cU_q}={\sf w}_q+{\sf R}_q=\left.\fT_q^-\right|_{\cU_q}.
\een
Thus, it follows that

\ben
e^{\bb p}\hat{\fT}_q(-p)=\hat{\fT}_q(p).
\een
That eq. \eqref{eq:clustrem2} is fulfilled follows directly from eq. \eqref{eq:LKMSN4} by which the proof of $(ii)\Rightarrow(iii)$ is complete.

It remains to show $(iii)\Rightarrow (i)$.  From (\ref{eq:LKMSNfou}) it follows together with the canonical commutation relations (\ref{eq:CCR}) that
\begin{align}
i\hat{E}(p)&=(e^{\bb p}-1)\hat{\fT}_q(-p),
\end{align}
where the symmetry of ${\sf R}_q$ has been used. Again, the distributional solution $\fw_q(-p)$ of the latter equation is given by
\ben
\hat{\fT}_q(-p)=\frac{i\hat{E}(p)}{e^{\bb p}-1}+c\delta(p),\quad c\in\IR.
\een 
Thus, by (\ref{eq:KMStwo}) there exists a $\bb$-KMS state $\o_\bb$, such that
\ben
\hat{\fT}_q(p)=\wfou_{\bb}(p)+c\delta(p).
\een
In view of the local clustering property, eq. (\ref{eq:clustrem2}), we find that $\hat{\fT}_q$ cannot contain a constant term, i.e. $c=0$ and thus
\ben
\hat{\fT}_q=\fw_{\bb}
\een
in the sense of distributions. This, together with the properties of ${\sf R}_q$, is equivalent to the $[\bb,q,N]$-LTE condition by Corollary \ref{coro:LTENequiv}, which completes the proof of the theorem. 
\end{proof}

\section{LTE States with Mixed Temperature}

So far we have only discussed locally thermal states with sharply defined thermal parameters. A generalization of this arises if one enlarges the space of thermal reference states to the space $C_B$ of mixtures of KMS states, as discussed in Section 2. 

\begin{defi}
Let $\rho$ be a positive normalized measure with support in some compact $B\subset V_+$. We shall call the state $\o_{B}$ on $\cA(\IM)$, defined by
\ben
\o_{B}(A)=\int_{B}d\rho(\bb)\o_\bb(A), \quad A\in\cA(\IM),
\een 
the \emph{mixed-temperature state with respect to $\rho$}.
\end{defi}
 
 It follows immediately from the respective properties of the KMS states $\o_\bb$ that $\o_B$ is an (not necessarily quasifree) analytic Hadamard state, i.e. it fulfills eq. \eqref{eq:WFAmink}. Thus, for any point $q\in\IM$ we can compute the relative-variable two-point function $\fw_q^B\in\cD'(\IR^4)$, according to Lemma \ref{lemma:pullback}. Since $\o_B$ can also easily be seen to be spacetime-translation invariant, $\o_B\circ\t_{(1,a)}=\o_B,a\in\IR^4$, the relative-variable two-point function does not depend on the choice of the point $q\in\IM$, i.e. $\fw_q^B\equiv\fw_B$. Furthermore we have $\fw_B\in\cS'(\IR^4)$ since the KMS two-point functions $\fw_\bb$ are tempered as well.

\begin{defi}
\label{defi:mLTE}
 Let $\omega$ be a Hadamard state on $\cA(\cM)$ and $N\in\IN$. We say that $\o$ is a \emph{local thermal equilibrium state of order $N$ at $q\in\cM$ with mixed temperature distribution $\rho$}, or $[\rho,q,N]$-\emph{LTE state} for short, iff there exists a positive normalized measure $\rho$ with support in some compact $B\subset V_+(q)$,  such that
 \ben
 \o(s(q))=\o_{B}(s(q)),\quad \forall s(q)\in S_q^n,\ n\leq N,
\label{eq:mLTE}
 \een
 where $\o_{B}$ is the mixed temperature state with respect to $\rho$.
 
\noindent We say that $\o$ is a $[\rho,q]$\emph{-LTE state} iff it is a $[\rho,q,N]$-LTE state for all $N\in\IN$. 
\end{defi}

It has been shown in \cite{BOR02} that for given $q\in\cM$ and $N\in\IN$ there exist non-trivial $[\rho,q,N]$-LTE states. This result was extended by Solveen \cite{Sol10} as follows:  For any compact subset $\cO$ of Minkowski spacetime and finite $N\in\IN$ there exist probability measure-valued functions $q\mapsto\rho_q,\ q\in\cO$ such that there are states which are non-trivial $[\rho_q,q,N]$-LTE states for all $q\in\cO$.

\subsubsection*{Mixed-temperature LTE states of infinite order}

We first show, in analogy to the case of sharp temperature, that the $[\rho,q]$-LTE condition can be expressed in terms of the relative-variable two-point function ${\sf w}_q\in\cD'(\cU_q)$. As before, we denote by $\cU_q$ an open neighbourhood of $0\in\IR^4$ with $\cU_q\subset\{z\in\IR^4:q\pm z/2\in\cM\}$.

\begin{lemma}
\label{lemma:mLTEequiv}
 Let $q\in\cM$. A Hadamard state $\o$ on $\cA(\cM)$ is a $[\rho,q]$-LTE state if and only if there exists a mixed-temperature state $\o_{B}$ with respect to $\rho$ on $\cA(\cM)$, such that 
\ben
[\del_\va ({\sf w}_{B}-{\sf w}_q)](0)=0 \quad\forall \va\in\IN_0^4,
\label{eq:mLTEequiv}
\een
where ${\sf w}_B\in\cD'(\cU_q)$ is the relative variable two-point function of $\left.\o_B\right|_\cM$.
\end{lemma}

\begin{proof}
The proof is verbatim the same as the proof of Lemma \ref{lemma:LTEequiv}, with $\o_\bb$ replaced by $\o_B$.
\end{proof}

If, in addition, the state $\o$ fulfills the analytic Hadamard condition we have the following Corollary which also is proven by analogous arguments as in the sharp-temperature case, replacing $\o_\bb$ by $\o_B$.

\begin{coro}
\label{coro:mLTEequiv}
Let $q\in\cM$. An analytic Hadamard state $\o$ on $\cA(\cM)$ fulfills the $[\rho,q]$-LTE condition if and only if there exists a mixed-temperature state $\o_B$ with respect to $\rho$ on $\cA(\IM)$, such that 
\ben
\fw_q=\fw_B
\een
in the sense of distributions, where ${\sf w}_B\in\cD'(\cU_q)$ is the relative-variable two-point function of $\left.\o_B\right|_\cM$.
\end{coro}

Now we aim at an intrinsic description of mixed LTE states by defining a class of states which, in a sense, can be regarded as mixtures of LKMS states.

\begin{defi}
 Let $q\in\IM$ and let $\o$ be an analytic Hadamard state on $\cA(\cM)$. Then we say that $\o$ fulfills the \emph{mixed-temperature local KMS condition at $q$ with respect to $\rho$}, or \emph{$[\rho,q]$-LKMS condition} for short, iff there exists a positive normalized measure $\rho$ with support in some compact $B\subset V_+(q)$, such that
\ben
{\sf w}_q=\int_B d \rho(\tilde{\bb}) {\sf w}_q^{\tilde{\bb}}
\een
in the sense of distributions. Here, ${\sf w}_q^{\tilde{\bb}}\in\cD'(\cU_q)$ are the relative-variable two-point functions at $q$ of states $\o_{\tilde{\bb},q}$ on $\cA(\cM)$ which fulfill the $[\tilde{\bb},q]$-LKMS condition. 
\end{defi}

In fact, it is not hard to see that this condition is indeed equivalent to the $[\rho,q]$-LTE condition given above.

\begin{theorem}
\label{theorem:mLKMSequiv}
 Let $q\in\cM$. An analytic Hadamard state $\o$ on $\cA(\cM)$ is a $[\rho,q]$-LTE state if and only if $\o$ fulfills the $[\rho,q]$-LKMS condition.
\end{theorem}

\begin{proof}
Assume that $\o$ fulfills the $[\rho,q]$-LKMS condition. By the $[\tilde{\bb},q]$-LKMS condition on the states $\o_{\tilde{\bb},q}$ it follows that
\ben
{\sf w}_q=\int_B d \rho(\tilde{\bb}) {\sf w}_q^{\tilde{\bb}}=\int_B d \rho(\tilde{\bb}) {\sf w}_{\tilde{\bb}}={\sf w}_B.
\een
Therefore, $\o$ is a $[\rho,q]$-LTE state by Corollary \ref{coro:mLTEequiv}.

 Assume conversely that $\o$ is a $[\rho,q]$-LTE state. From Corollary \ref{coro:mLTEequiv} it follows that there exists a mixed-temperature state $\o_{B}$ such that
\ben
{\sf w}_q={\sf w}_{B}=\int_B d\rho(\tilde{\bb}) {\sf w}_{\tilde{\bb}}
\label{eq:mLTEproof}
\een
in the sense of distributions. By Corollary \ref{coro:exLKMS}, we have for any given $\tilde{\bb}\in B$ an analytic Hadamard state $\o_{\tilde{\bb},q}$ on $\cA(\cM)$ which fulfills the $[\tilde{\bb},q]$-LKMS condition, i.e. ${\sf w}_q^{\tilde{\bb}}={\sf w}_{\tilde{\bb}}$. We can thus replace the relative variable two-point functions ${\sf w}_\bb$ occurring in eq. \eqref{eq:mLTEproof} by the relative variable-two-point functions ${\sf w}_q^{\tilde{\bb}}$ which yields
\ben
{\sf w}_q={\sf w}_{B}=\int_B d\rho(\tilde{\bb}) {\sf w}_{\tilde{\bb}}=\int_B d\rho(\tilde{\bb}){\sf w}_q^{\tilde{\bb}}.
\een
This shows that the $[\rho,q]$-LTE condition implies the $[\rho,q]$-LKMS condition and the proof is complete.
\end{proof}

\subsection*{Mixed-temperature LTE states of finite order}

\begin{lemma}
\label{lemma:mLTENequiv}
Let $q\in\cM$. A Hadamard state $\o$ on $\cA(\cM)$ is a $[\rho,q,N]$-LTE state if and only if there exists a mixed-temperature state $\o_B$ with respect to $\rho$ on $\cA(\IM)$ such that
\ben
[\del_\va ({\sf w}_{B}-{\sf w}_q)](0)=0 \quad\forall \va\in\{\va'\in\IN_0^4:\abs{\va'}\leq N\},
\een
where ${\sf w}_B\in\cD'(\cU_q)$ is the relative variable two-point function of $\left.\o_B\right|_\cM$.
\end{lemma}

\begin{proof}
The proof is verbatim the same as the proof of Lemma \ref{lemma:LTENequiv}, with $\o_\bb$ replaced by $\o_B$.
\end{proof}

Again, assuming the analytic Hadamard condition to be fulfilled by $\o$ we have the following Corollary.

\begin{coro}
\label{coro:mLTENequiv}
Let $q\in\cM$ and $N\in\IN_0$. An analytic Hadamard state $\o$ on $\cA(\cM)$ fulfills the $[\rho,q,N]$-LTE condition if and only if there exists a mixed-temperature state $\o_B$ with respect to $\rho$ on $\cA(\IM)$ and a symmetric distribution ${\sf R}_q\in\cD'(\cU_q)$ with $\WF_A({\sf R}_q)=\emptyset$ and
\ben
[\del^\va {\sf R}_q(0)]=0 \quad\forall \va\in\{\va'\in\IN_0^4:\abs{\va'}\leq N\}
\een 
such that
\ben
{\sf w}_q={\sf R}_q+\int_B d \rho(\tilde{\bb}) {\sf w}_q^{\tilde{\bb}}
\een
in the sense of distributions, where ${\sf w}_B\in\cD'(\cU_q)$ is the relative variable two-point function of $\left.\o_B\right|_\cM$.
\end{coro}

\begin{defi}
\label{defi:mLKMS}
 Let $q\in\IM$, $N\in\IN_0$ and $\o$ an analytic Hadamard state on $\cA(\cM)$. Then we say that $\o$ fulfills the \emph{mixed-temperature local KMS condition at $q$ with respect to $\rho$ at order $N$}, or \emph{$[\rho,q,N]$-LKMS condition} for short, iff there exists a positive normalized measure $\rho$ with support in some compact $B\subset V_+(q)$ and a symmetric distribution ${\sf R}_q\in\cD'(\cU_q)$ with $\WF_A({\sf R}_q)=\emptyset$ and
\ben
[\del^\va {\sf R}_q(0)]=0 \quad\forall \va\in\{\va'\in\IN_0^4:\abs{\va'}\leq N\}
\een 
such that
\ben
{\sf w}_q={\sf R}_q+\int_B d \rho(\tilde{\bb}) {\sf w}_q^{\tilde{\bb}}
\een
in the sense of distributions. Here, ${\sf w}_q^{\tilde{\bb}}\in\cD'(\cU_q)$ are the relative-variable two-point functions at $q$ of states $\o_{\tilde{\bb},q}$ on $\cA(\cM)$ which fulfill the $[\tilde{\bb},q]$-LKMS condition. 
\end{defi}

\begin{theorem}
Let $q\in\cM$. An analytic Hadamard state $\o$ on $\cA(\cM)$ is a $[\rho,q,N]$-LTE state if and only if $\o$ fulfills the $[\rho,q,N]$-LKMS condition.
\end{theorem}

\begin{proof}
This is proven by analogous arguments an in the proof of Theorem \ref{theorem:mLKMSequiv}, using Corollary \ref{coro:mLTENequiv} and the existence result on $[\bb,q]$-LKMS states, Corollary \ref{coro:exLKMS}.
\end{proof}

\section{Summary and outlook}

A new proposal for a concept of local temperature in quantum field theory has been made and
investigated in the present work. The basis of the concept is a local version of the KMS
condition at the level of the two-point correlation of states of a quantum field. That local
version permits temperature variations from point to point in spacetime, as well as variations
of the time-direction with respect to which the condition is formulated and which ought to be
seen as defining an (approximate) momentary equilibrium rest-frame. The main result of the present work is the equivalence of our LKMS condition and the LTE condition of Buchholz, Ojima and Roos for states of the quantized Klein-Gordon field which fulfill the analytic microlocal spectrum condition, on patches of Minkowski spacetime. This result indicates the utility of our LKMS condition in view of the sound physical motivation of the LTE condition in providing a characterization of states in quantum field theory to which one can ascribe a locally
varying temperature. On the other hand, the result sheds a new light on the LTE condition, indicating
its relation to a local version of the KMS condition.

Nevertheless, the LTE condition is quite restrictive and its is likely that the set of LTE resp. LKMS states (to infinite order) that have a varying temperature distribution might be very small. This applies in particular to situations where the quantum fields propagate in the presence of time-varying external background fields; yet that is a potential domain of application for LTE and LKMS states, as for example in quantum field theory in the early cosmological eras. Therefore, seeking generalizations to the LTE concept appears a promising task. 

We notice that the LKMS condition proposed in this paper does not involve any regularization nor any Hadamard subtraction. For this reason, this condition might be easier to generalize to the case of a quantum field propagating on a curved spacetime than the LTE condition. In particular, a generalization to the curved case of the LKMS condition of order zero, as introduced in Definition \ref{defi:LTEN}, is surely possible. In that case, the function ${\sf R}_q$ introduced in Corollary \ref{coro:LTENequiv}, or in Definition \ref{defi:LTEN}, quantifies the failure of a state of being an exact LKMS state. 

More generally, we think that the LKMS condition offers several possibilities for a generalization. As already pointed out in the beginning, without imposing the analytic microlocal spectrum condition, the LKMS condition is implied by the LTE condition, but not necessarily vice versa. The finite order LKMS conditions of Definition \ref{defi:LTEN} could be modified to yield a generalized LKMS condition in a suitable microlocal sense which might render itself useful to quantum field theory in curved spacetime, and we see such a line of research as promising
in the future. Furthermore, the LKMS condition (or suitable generalizations) could provide links to other
conditions on states in quantum field theory in curved spacetimes which express thermodynamic stability or
asymptotic relation to thermal states, see, e.g., \cite{FV03, DHP11}. There is also the possibility to investigate
the LKMS condition beyond two-point correlations. The operator product expansion \cite{HK12} could provide a useful tool in that context. Progress in this direction should make it possible to understand potential relations between the LKMS condition and other characterizations of (locally) thermal states in interacting quantum field theory \cite{ABDM09, FL13}.

\subsubsection*{Acknowledgements} Part of this work has been carried out during a stay of the authors at the ESI, Vienna at the workshop ``AQFT - Its status and its future" in May 2014. M.G. gratefully acknowledges financial support by the Max Planck Institute for Mathematics in the Sciences and its International Max Planck Research School (IMPRS) ``Mathematics in the Sciences".

\newpage
\appendix
\section{Notation and Conventions}

The metric convention is $\eta=\text{diag}(+,-,-,-)$. We use the Einstein summation convention.
The open forward  light cone $V_+$ is defined as
\ben
V_+:=\{e\in T\IM:e_\mu e^\mu>0, e_0>0\},
\een
and its boundary is denoted by $\del V_+$. The open backward light cone is defined by $V_-:=-V_+$ and $\del V_-$ is its boundary. The set of unit vectors in $V_+$, also called the set of \emph{time-directions}, is denoted by 
\ben
V_+^1:=\{e\in V_+:e^\mu e_\mu=1\}.
\een
The open forward (backward) light cone emanating from a point $q\in\IM$ is defined by 
\ben
V_\pm(q):=\{e\in\IR^4: e-q\in V_\pm\}=V_\pm + q
\een
and its boundary by $\del V_\pm(q)$.

We symbolically write  $\cF[h(x)](p)$ for the Fourier transform of a function $h$ with respect to the variable $x\in\IR^n$. This symbolic notation also applies if $h$ is a (tempered) distribution.  Our conventions for the one- and four-dimensional Fourier transforms are as follows:
\begin{align}
\cF[f(x)](p)\equiv\hat{f}(p)&:=\frac{1}{(2\pi)^2}\int_{\IR^4}\diff^4 x\ f(x)e^{i x p }\\
\cF[f(t)](k)\equiv\hat{f}(k)&:=\frac{1}{\sqrt{2\pi}}\int_{\IR}\diff t\ f(t)e^{itk}
\end{align}
The unique vacuum state on the Klein-Gordon algebra $\cA(\IM)$ is denoted by $\o_\text{vac}$.
The set of real-analytic functions on a subset $X\subset\IR^n$ is denoted by $C^A(X)$.

\section{Pullbacks}
Let $\bb=\b e$, with $\b>0$ and $e\in V_{+,1}$, and $\o_\bb$ a $\bb$-KMS state on $\cA(\IM)$.  In this appendix we give the relative-time-variable two-point function $\fu_\bb$ and the relative-variable two-point function $\fw_\bb$ of $\o_\bb$. Those are well-defined as distributions in $\cS'(\IR)$ and $\cS'(\IR^4)$ by Lemma (\ref{lemma:pullback}).

From eq. \eqref{eq:KMStwo} one obtains that the relative-variable two-point function $\fw_{\bb}$ of $\o_\bb$ is given by
\ben
\fw_\bb(z )=\frac{1}{(2\pi)^3}\int d^4p\ \frac{\eps(p_0)\d(p^2-m^2)}{1-e^{\bb p}}e^{-ipz},
\een
which yields
\ben
\wfou_\bb(p)=\frac{1}{2\pi}\cdot\frac{\eps(p_0)\d(p^2-m^2)}{1-e^{\bb p}}.
\een
In order to compute $\fu_\bb\in\cS'(\IR)$, which arises as the restriction of the relative-variable two-point function to $\IR\cdot e$, we switch to the coordinate frame in which $e=(1,\vec{0})$ and obtain the formal expression
\ben
\fu_\bb(t)=\frac{1}{(2\pi)^3}\int d^4p\ \frac{\eps(p_0)\d(p^2-m^2)}{1-e^{\b p_0}}e^{-ip_0 t}.
\een
We integrate over the $p_0$-component, switch to spherical coordinates with $r=\abs{\vec{p}}$ and then substitute $k=\sqrt{r^2+m^2}$, which yields
\begin{align}
 \fu_\bb(t)&=\frac{1}{(2\pi)^3}\int\limits_{\IR^3}\frac{d^3\vec{p}}{2\o_{\vec{p}}}\left(\frac{e^{-i\o_{\vec{p}} t}}{1-e^{-\b\o_{\vec{p}}}}-\frac{e^{i\o_{\vec{p}} t}}{1-e^{\b\o_{\vec{p}}}}\right)\non\\
&=\frac{1}{4\pi^2}\int\limits_0^\infty dr\ \frac{r^2}{\sqrt{r^2+m^2}}\left(\frac{e^{-i\sqrt{r^2+m^2}t}}{1-e^{\b\sqrt{r^2+m^2}}}-\frac{e^{i\sqrt{r^2+m^2}t}}{1-e^{\b\sqrt{r^2+m^2}}}\right)\non\\
&=\frac{1}{4\pi^2}\int\limits_m^\infty dk\ \sqrt{k^2-m^2}\left(\frac{e^{-ikt}}{1-e^{-\b k}}-\frac{e^{ikt}}{1-e^{\b k}}\right)\non\\
&=\frac{1}{4\pi^2}\cdot\left[\int\limits_m^\infty dk\ \frac{\sqrt{k^2-m^2}}{1-e^{-\b k}}e^{-ikt}-\int\limits_{-\infty}^m dk \frac{\sqrt{k^2-m^2}}{1-e^{-\b k}}e^{-ikt}\right]\non\\
&=\frac{1}{4\pi^2}\int\limits_{-\infty}^\infty dk\ \eps(k)\Th(k^2-m^2)\frac{\sqrt{k^2-m^2}}{1-e^{-\b k}} e^{-ikt}\\
&=\frac{1}{2\sqrt{2}\pi^{3/2}}\cdot\cF^{-1}[\ufou_\bb(k)](t).
\end{align}
Therefore, the Fourier transform $\ufou_\bb$ of $\fu_\bb$ is given by
\ben
\ufou_\bb(k)=\frac{1}{2\sqrt{2}\pi^{3/2}}\ \frac{\eps(k)\Th(k^2-m^2)\sqrt{k^2-m^2}}{1-e^{-\b k}},
\label{eq:KMSpullm}
\een
which reduces in the massless case with $m=0$ to
\ben
\ufou_\bb(k)=\frac{1}{2\sqrt{2}\pi^{3/2}}\ \frac{k}{1-e^{-\b k}}.
\label{eq:KMSpull0}
\een
The relative-time-variable commutator distribution $\fE\in\cD'(\IR)$, is therefore given by
\ben
\hat{\fE}(k)=(-i)\cdot(\ufou_\bb(k)-\ufou_\bb(-k))=-\frac{i}{2\sqrt{2}\pi^{3/2}}\eps(k)\Th(k^2-m^2)\sqrt{k^2-m^2},
\label{eq:commpullm}
\een
which reduces for $m=0$ to 
\ben
\hat{\fE}(k)=-\frac{i}{2\sqrt{2}\pi^{3/2}}k.
\label{eq:commpull0}
\een
We thus obtain the following remnant of relation (\ref{eq:KMStwo}) in the massless  ($m=0$) as well as in the massive ($m>0$) case:
\ben
\ufou_\bb(k)=\frac{i\hat{\fE}(k)}{1-e^{-\b k}}.
\label{eq:KMSpull}
\een

\section{Proofs of some lemmas}

\begin{proof}[\textbf{Proof of Lemma \ref{lemma:pullback}}]
This is an application of H\"ormander's criterion. The set of normals of the map $\c_{(q,e)}\equiv(\c_1,\c_2)$ is defined as
\ben
 N_\c=\{(\c_{(q,e)}(t),\eta)\in (\cM \times\cM)\times (\IR^4\times\IR^4)|(\c_{(q,e)}^\prime)^T(t)[\eta]=0\}.
\een
We have $(\c_{(q,e)}^\prime)^T(t)=\frac{1}{2}(-e,e)$ for all $t\in\cI_{q,e}$ and thus
\ben
(\c_{(q,e)}^\prime)^T(t)[\eta]= \frac{1}{2}(\eta_2^\mu e_\mu-\eta_1^\mu e_\mu).
\label{eq:chiT}
\een   
Now suppose that $(x,p;x',p')\in\WF(\o_2)\cap N_\c$. In, particular this implies that $p$ is lightlike and future-pointing and $p'=-p$. But from \eqref{eq:chiT} it follows that we have $p^\mu e_\mu=0$, from which we we conclude that $\WF(\o_2)\cap N_\chi=\emptyset$ since $p^\mu e_\mu\neq 0$ for all future-pointing lightlike $p$ and timelike $e$. We can thus apply \cite[Thm 8.2.4]{Hoe90}, by which $\chi_{q,e}^*\o_2\in\cD'(\cI_{q,e})$ is well defined and  eq. \eqref{eq:WFu} holds. In the analytic case we can apply \cite[Thm 8.5.1]{Hoe90} to obtain eq. \eqref{eq:WFAu}. Next, the set of normals to the map $\k_q$ is defined as 
$$ N_\k=\{(\k_q(\z),\eta)\in (\cM \times\cM)\times (\IR^4\times\IR^4):(\k'_q)^T(\z)[\eta]=0\}.$$
The differential $\k'_q$ is given by the respective Jacobi matrices so we have
\begin{align}
(\k'_q)^T(\z)[\eta]=\frac{1}{2}\begin{pmatrix}
                   -\eins_4,\eins_4
					      \end{pmatrix}[\eta]
=\frac{1}{2}(\eta_2-\eta_1),\non
\end{align}
which vanishes if and only if $\eta_1=\eta_2$. Thus, we have
$$ N_\k=\{(q-\z/2,q+\z/2,\eta,\eta):\z\in \cU_q \}$$ 
Comparing this with the wave front set of $\o_2$, Eqn. (\ref{eq:WFmink}), we immediately see that $N_\k\cap\WF(\o_2)=\emptyset$. We can thus again apply Theorem 8.2.4 of \cite{Hoe90}, by which $\k_q^*\o_2\in\cD'(\cU_q)$ is well defined and eq. \eqref{eq:WFw} holds. In the analytic case we can apply Theorem 8.5.1. of \cite{Hoe90}, which yields eq. \eqref{eq:WFAw} and completes the proof.
\end{proof}

In the proof of Lemma \ref{lemma:LTEequiv} we will make use of the following polarization identity for totally symmetric multilinear maps.

\begin{lemma}
\label{lemma:polid}
Let $E$ and $F$ be $\IR$-vector spaces and $T:E^n\to F$ a totally symmetric multilinear map. If we denote the restriction to the diagonal by $\tilde{t}(v):=T(v,\ldots,v)$, for $v\in E$, there holds the polarization identity
\ben
T(v_1,\ldots,v_n)=\frac{1}{n!}\sum\limits_{k=1}^n (-1)^{n-k}\sum\limits_{J,\abs{J}=k}\tilde{t}\left(\sum\limits_{i\in J} v_i\right)
\een
 where the $J$ are subsets of $\{1,2,\ldots,n\}$ and $\abs{J}$ denotes the number of elements of $J$.
\end{lemma}
\begin{proof}
 A nice proof can be found e.g. in \cite{Tho14}.
\end{proof}

This leads to a slight generalization of Lemma 5.1.2. in \cite{Kri99}.

\begin{lemma}
\label{lemma:symtens}
 Any symmetric $\binom{0}{n}$-tensor $T$ is uniquely determined by its values $T(e,\ldots,e)$ for all timelike vectors $e$ with $
e^2=e^\mu e_\mu=1$. Similarly, $T$ is uniquely determined by its values $T(,\ldots,f)$ for all spacelike vectors $f$ with $f^2=f^\mu f_\mu=-1$ .
\end{lemma}

\begin{proof}
We give the proof for the timelike case. The spacelike case is proven analogously. Let $T$ and $S$ be two symmetric $\binom{0}{n}$-tensors which coincide on the diagonal for all timelike unit vectors, i.e. $T(e,\ldots,e)=S(e,\ldots,e)$ for all $e$ with $e^2=1$. Let now $v$ be a timelike vector. Denoting the normalized vector by $\overline{v}=\frac{v}{\sqrt{v^2}}$, we have
\ben
T(v,\ldots,v)=\left(\frac{1}{\sqrt{v^2}}\right)^n T(\overline{v},\ldots,\overline{v})=\left(\frac{1}{\sqrt{v^2}}\right)^n S(\overline{v},\ldots,\overline{v})=S(v,\ldots,v)
\een
Now, for any arbitrary vector $w$ we can find a $\d>0$ such that $v+tw$ is timelike as long as $t\in[-\d,\d]$, since $V_+\cup V_-$ is an open set. This implies that $T$ and $S$ coincide on the diagonal for any vector $w$, since
\begin{align}
T(w,\ldots,w)&= \frac{1}{n!}\dnull{n} T(v+tw,\ldots,v+tw)\notag\\
		& = \frac{1}{n!}\dnull{n} S(v+tw,\ldots,v+tw)=S(w,\ldots,w).
\end{align}
Hence, by Lemma \ref{lemma:polid}, $T$ and $S$ coincide on any vectors and thus $T=S$. 
\end{proof}

\begin{proof}[\textbf{Proof of Lemma \ref{lemma:KMSrem4}}]
By eq. (\ref{eq:KMStwo}), the Fourier transform of $\fw_\bb\in\cS'(\IR^4)$ is given by
\ben
\wfou_\bb(p)=\frac{1}{2\pi}\cdot\frac{i \hat{E}(p)}{1-e^{-\bb p}},
\een
where $E\in\cS'(\IR^4)$ denotes the commutator distribution. In view of the time-clustering property of $\o_\bb$, eq. (\ref{eq:clust}), this is equivalent to 
\ben
e^{\bb p}\wfou_\bb(-p)=\wfou_\bb(p).
\label{eq:KMSfou}
\een
 Thus, defining $$\G_\bb=\{y\in\IR^4:e^{yp}\wfou_\bb(-p)\in\cS'(\IR^4)\},$$ we have $0\in\G_\bb$ and $\bb\in\G_\bb$.  Since the set $\G_\bb$ is convex \cite[Thm. 2-5]{SW00} this implies $\l\bb\in\G_\bb$ for all $\l\in[0,1]$, or equivalently
\ben
e^{\s e p}\wfou_\bb(-p)\in\cS'(\IR^4),\quad \forall \s\in[0,\b].
\een
Applying Theorem 7.4.2 in \cite{Hoe90}, one obtains that the function $F_\bb$, defined by $$F_\bb(z +i\s e):=\cF[e^{\s e p}\wfou_q(-p)](z )\quad\forall z \in\IR^4,\s\in(0,\b),$$ fulfills properties $(i)$ and $(ii)$ above. In particular we have $F(\Bigcdot +i\s_0e)\in\cS'(\IR^4)$ for every fixed $\s_0\in(0,\b)$. Since the Fourier transform $\cF:\cS'\to\cS'$ is weakly continuous and
 \ben
e^{\s e p}\wfou_\bb(-p) \xrightarrow[\s\to 0^+]{} \wfou_\bb(-p)
\een
 in the weak sense, we find
\ben
F_\bb(z+i\s e)=\cF[e^{\s e p}\wfou_\bb(-p)] \xrightarrow[\s\to 0^+]{} \cF[\wfou_\bb(-p)]=\fw_\bb(z).
\een
Now consider the function $G_\bb(z +i\eta e):=F_\bb(z +i(\b-\eta)e)$, which is also holomorphic on $\cT_\bb$. Then analogous arguments as for $F_\bb$, together with (\ref{eq:KMSfou}) yield
\ben
 G_\bb(z +i\eta e)=\cF[e^{(\b-\eta) p e}\wfou_\bb(-p)] \xrightarrow[\eta\to 0^+]{} \cF[e^{\bb p}\wfou_\b(-p)]=\cF[\wfou_\b(p)]=\fw_\bb(-z ),
\een
which completes the proof.
\end{proof}

%

\newpage
\bibliographystyle{amsalpha}

\providecommand{\bysame}{\leavevmode\hbox to3em{\hrulefill}\thinspace}
\providecommand{\MR}{\relax\ifhmode\unskip\space\fi MR }
\providecommand{\MRhref}[2]{%
  \href{http://www.ams.org/mathscinet-getitem?mr=#1}{#2}
}
\providecommand{\href}[2]{#2}

\end{document}